\documentclass[journal, final, onecolumn, 10pt]{IEEEtran}
\usepackage[nocompress]{cite}

\usepackage{subfigure}
\usepackage{float}

\usepackage[bookmarks=false, hidelinks]{hyperref}

\usepackage{amsmath,amssymb,amsfonts}
\allowdisplaybreaks
\usepackage{bm}                         
\usepackage{multirow}
\usepackage{cases}                       
\usepackage{tikz}
\usetikzlibrary{arrows.meta, positioning, automata, shapes, decorations.pathreplacing, calc, shadows.blur}

\usepackage[normalem]{ulem}             

\usepackage{amsthm}
\newtheorem{theorem}{Theorem}[section]
\newtheorem{lemma}[theorem]{Lemma}
\newtheorem{corollary}[theorem]{Corollary}
\newtheorem{remark}{Remark}

\usepackage[linesnumbered,noend,ruled,vlined]{algorithm2e}
\usepackage{algpseudocode}

\usepackage{graphicx}
\usepackage{textcomp}
\usepackage{xcolor}

\usepackage{balance}

\begin{document}

\title{Restless Multi-Process Multi-Armed Bandits\\ with Applications to Self-Driving Microscopies
\thanks{
}
}



\author{
\IEEEauthorblockN{
Jaume Anguera Peris\IEEEauthorrefmark{1}, 
Songtao Cheng\IEEEauthorrefmark{2}, 
Hanzhao Zhang\IEEEauthorrefmark{3}, 
Wei Ouyang\IEEEauthorrefmark{2}, 
Joakim Jaldén\IEEEauthorrefmark{1}}\\
\IEEEauthorblockA{
\IEEEauthorrefmark{1}Department of Information Science and Engineering\\
\IEEEauthorrefmark{2}Department of Applied Physics\\
\IEEEauthorrefmark{3}Department of Biophysics\\
KTH Royal Institute of Technology, Stockholm, Sweden\\
\{jaumeap, jalden\}@kth.se, \{songtao.chen, hanzhao.zhang, wei.ouyang\}@scilifelab.se}
}

\maketitle
\begin{abstract}
High-content screening microscopy generates large amounts of live-cell imaging data, yet its potential remains constrained by the inability to determine \textit{when} and \textit{where} to image most effectively. Optimally balancing acquisition time, computational capacity, and photobleaching budgets across thousands of dynamically evolving regions of interest still remains an open challenge, further complicated by limited field-of-view adjustments and sensor sensitivity. Existing approaches either rely on static sampling or heuristics that neglect the dynamic evolution of biological processes, leading to inefficiencies and missed events. Here, we introduce the restless multi-process multi-armed bandit (RMPMAB), a new decision-theoretic framework in which each experimental region is modeled not as a single process but as an ensemble of Markov chains, thereby capturing the inherent heterogeneity of biological systems such as asynchronous cell cycles and heterogeneous drug responses. Building upon this foundation, we derive closed-form expressions for transient and asymptotic behaviors of aggregated processes, and design scalable Whittle index policies with sub-linear complexity in the number of imaging regions. Through both simulations and a real biological live-cell imaging dataset, we show that our approach achieves substantial improvements in throughput under resource constraints. Notably, our algorithm outperforms Thomson Sampling, Bayesian UCB, $\epsilon$-Greedy, and Round Robin by reducing cumulative regret by more than $37\%$ in simulations and capturing $93\%$ more biologically relevant events in live imaging experiments, underscoring its potential for transformative smart microscopy. Beyond improving experimental efficiency, the RMPMAB framework unifies stochastic decision theory with optimal autonomous microscopy control, offering a principled approach to accelerate discovery across multidisciplinary sciences.
\end{abstract}


\section{Introduction}
\label{sec:introduction}
High content screening (HCS) microscopy has revolutionized live cell imaging by enabling the monitoring of cellular responses across thousands of regions of interest. Despite this capability, many critical biological events, such as transient drug effects or rare cellular transitions, frequently go unobserved due to inherent limitations in the imaging process. Constraints imposed by slow hardware navigation between fields of view, photobleaching and phototoxicity, the adjustment of imaging parameters, or the need to manage large numbers of samples in real time not only hinder progress but also introduce potential errors~\cite{mantz2025ai}. These challenges are exacerbated in smart microscopy, where automated workflows remain susceptible to variability, particularly when dealing with heterogeneous cell populations. Together, these limitations highlight the need for intelligent, autonomous systems capable of navigating the complexities of live-cell imaging under resource constraints~\cite{hinderling2025smart}.

To address these challenges, this paper focuses on large-scale and resource-constrained environments on HCS applications, where multi-well plates containing diverse cellular populations must be monitored for rare or transient events occurring asynchronously across different spatial locations~\cite{carreras2024artificial}. By integrating automated microscopy with advanced image analysis, HCS microscopes enable researchers to quantify and correlate drug effects on cellular events or targets by simultaneously measuring multiple signals from the same cell population, yielding data with a higher content of biological information than what is provided by single-target screens~\cite{zock2009applications}. However, the fundamental challenge in autonomous microscopy remains optimal resource allocation, as limited imaging time, computational resources, and photobleaching budgets require autonomous systems to make intelligent decisions about which regions to observe while biological processes evolve continuously and independently across thousands of samples~\cite{tom2024self}.

The resource allocation challenge in microscopy naturally maps to the multi-armed bandit (MAB) problem from decision theory, where an agent must choose between multiple \textit{arms} with unknown and potentially changing \textit{rewards}. In microscopy contexts, each potential imaging location (field of view, region of interest, or well) represents an arm, and the reward corresponds to the biological information gained from the observation. While classical MAB formulations have shown promise in adaptive sampling and medical imaging, they assume static arm states and fail to account for the stochastic evolution of live cell systems. The restless multi-armed bandit (RMAB) framework addresses this critical limitation by incorporating temporal dynamics, allowing arm states to evolve stochastically regardless of observation~\cite{papadimitriou1999complexity}. In biological applications, this mirrors the reality that cells persist in proliferating, differentiating, or responding to stimuli whether or not they are being imaged. Despite this natural fit between RMAB assumptions and the constraints of biological imaging, its application to smart microscopy remains largely unexplored.

This paper advances the field by proposing a new RMAB formulation tailored to biological imaging. Unlike traditional RMAB models that represent each arm with a single Markov process, we introduce a framework for restless multi-armed bandits in which each arm encapsulates multiple Markov processes. This extension better captures the heterogeneity inherent in biological systems, such as varying activity levels across cells within a single imaging region, different cell cycle phases within a population, or spatially distributed cellular responses to drug treatments. Our framework bridges HCS with RL-based decision theory, providing a principled approach to prioritize regions most likely to yield valuable biological information. In particular, our work makes three key contributions that span theory, algorithms, and applications. First, from a theoretical perspective, we establish a rigorous mathematical framework for analyzing systems of multiple Markov processes, deriving both exact and asymptotic expressions for their collective dynamics. These results extend beyond bandit problems to any application involving aggregated stochastic processes. Second, from an algorithmic standpoint, we derive computationally tractable and scalable Whittle index policies that maintain linear complexity advantages while extending to multi-process settings, enabling real-time decision-making for systems with thousands of observable regions. Third, from an applications perspective, we demonstrate substantial empirical improvement in biological imaging, validated through both simulated environments and experimental stem cell profiling data.

Overall, our novel approach not only ensures equitable and optimal resource allocation but also achieves scalability and precision, accelerating insights into regenerative medicine and cellular dynamics by minimizing observational overhead while maximizing information yield.
Besides, the implications of this manuscript extend beyond microscopy to the broader vision of self-driving laboratories~\cite{abolhasani2023rise}, where autonomous systems must make intelligent decisions about resource allocation across multiple parallel experiments.

\subsection{Related work is optical microscopy}
The integration of machine learning (ML) has rapidly transformed the adaptive capabilities of modern optical microscopy systems. Initially, ML was leveraged primarily post-experiment to enhance data analysis, notably improving cell segmentation, pattern recognition, and phenotype classification to extract richer insights from acquired images. This post-hoc augmentation, while valuable, only supplements traditional workflows without fundamentally altering them~\cite{von2021democratising}. Subsequent real-time augmentation systems began assisting during acquisition through live segmentation and dimensionality reduction to aid operator decisions, yet they still required continuous human supervision for critical decision-making~\cite{mund2021ai}.

More recent developments in reinforcement learning (RL) have empowered microscopes with the ability to select measurement locations dynamically, adapting acquisition protocols in real time based on ongoing data streams and sample status without human intervention. The pioneering work in~\cite{kandel2023demonstration} demonstrated an AI-driven workflow for autonomous high-resolution scanning microscopy, coupling predictive modeling with on-the-fly scan optimization. Similar frameworks have been explored in event-driven microscopy, where neural detectors trigger acquisitions only when dynamic biological events are predicted to occur~\cite{mahecic2022event}. These systems reduce photobleaching and data redundancy but generally rely on threshold-based heuristics rather than provably optimal selection rules. Beyond these, many works address module-level autonomy in microscopy or platform integration for self-driving labs, but they assume that unobserved regions remain static until revisited, which is incompatible with live-cell dynamics where biological states evolve continuously~\cite{mund2021ai}. 

Broader perspectives on computationally steered microscopy envision the field converging toward fully autonomous, closed-loop systems. The forward-looking reviews in \cite{balasubramanian2023imagining} and \cite{carpenter2023smart} both highlight the need for scalable control algorithms that integrate sensing, inference, and actuation across multiple fields of view. However, these works remain conceptual and do not formalize resource allocation across dynamically evolving regions of interest.

Overall, the literature shows a rapid shift from ML-assisted analysis to adaptive acquisition, with the HCS and self-driving microscope communities both calling for principled, optimization-centric controllers \cite{shroff2024live}. In that sense, the theory behind RMAB offers the precise tools needed to allocate observations under evolving, partially-observed dynamics. The contribution of this work is therefore to unify smart and self-driving microscopy under a single decision-theoretic framework tailored to microscopy’s multi-process and resource-constrained environments.

\subsection{Related work in multi-armed bandits}
The field of MABs encompasses a diverse array of frameworks designed to address various sequential decision-making problems under uncertainty, with the most prominent being stochastic bandits, adversarial bandits, contextual bandits, and restless bandits~\cite{nino2023markovian}. Stochastic bandits assume that each arm's rewards are drawn from a fixed, unknown probability distribution, making them suitable for scenarios with stable underlying dynamics. Adversarial bandits, on the other hand, make no statistical assumptions about the reward generation process, allowing for potentially malicious or worst-case scenarios where an adversary may control the rewards. Contextual bandits extend the standard MAB setting by incorporating additional information (context) available to the agent before making each decision, enabling more nuanced decision-making strategies, and restless bandits further generalize the MAB problem by allowing the state of each arm to evolve over time, even when not selected. 

As discussed earlier, our proposed framework in smart microscopy necessitates a sophisticated model that can capture the stochastic nature of the cellular behavior within each well and account for the ongoing changes in unobserved wells. To address these challenges, we introduce a novel framework that we term the restless multi-process multi-armed bandit (RMPMAB). Our approach extends the traditional restless MAB model by incorporating multiple Markov processes within each arm, representing diverse dynamic behaviors within a field of view or a well. The restless nature of our proposed framework models the continuous evolution of cellular states in all regions of interest, mirroring the behavior of cells in the multi-well culture plate. The multi-process representation captures the heterogeneity of cell populations within each well, allowing for a more realistic representation of the complex biological system. And the multi-armed bandit structure enables the development of advanced decision-making strategies, allowing the model to adapt to the dynamic and evolving nature of biological systems in HCS experiments. Overall, while our framework may appear to have a simplified representation of cells through Markov processes, it ensures analytical tractability and effectively captures the core stochastic and heterogeneous properties of cellular dynamics in HCS environments.

It's worth noting why adversarial and contextual bandits are not incorporated into our framework, despite their unique strengths. Adversarial bandits, while robust to worst-case scenarios, assume that rewards can be arbitrarily controlled by an adversary. This assumption often leads to overly conservative strategies that fail to exploit the inherent structure and patterns present in biological systems. In our HCS scenario, the uncertainty and variability in cellular responses are governed by underlying biological mechanisms rather than adversarial manipulation. As a result, probabilistic models are better suited to capture the stochastic nature of cellular behavior, enabling us to leverage statistical patterns and prior knowledge about cellular dynamics. Moreover, contextual bandit models rely on the availability of external contextual information to influence decision-making. In the case of HCS, such contextual information would need to be sufficiently significant to impact the selection of wells. However, our framework considers the Markov processes to be the primary drivers of the cellular behavior in each region of interest. Since the state of the cells are fully observable only after imaging the well, external contextual information plays a limited role in our decision-making approach. By focusing on the stochastic and dynamic properties of the cellular systems, we ensure that our framework remains tailored to the specific challenges and requirements of HCS.

\subsection{Document organization}
The remainder of this manuscript is organized as follows. Section \ref{sec:markov_processes} establishes the mathematical framework for single and multiple independent Markov processes, analyzing their transient and asymptotic behaviors. Section \ref{sec:restless_MAB_framework} presents the RMAB formulation, derives the Whittle index policy, and discusses extensions to related problems. Section \ref{sec:simulation_results} evaluates our proposed policy through simulations against baselines such Round Robin, Bayesian UCB, and Thompson Sampling. Section \ref{sec:numerical_results} applies the policy to a biological dataset, demonstrating improved decision-making. Finally, Section \ref{sec:conclusions} summarizes our contributions and discusses broader implications for mathematical biology and computational decision-making.

\subsection{Notation}
Throughout this paper, sets are represented by calligraphic letters $\mathcal{S}$, and the cardinality (number of elements) of a finite set $\mathcal{S}$ is denoted by $|\mathcal{S}|$. Indexes are represented by lowercase letters $h$, samples (observations) of random variables are represented by uppercase letters $X$, and matrices are indicated by bold uppercase letters $\bm{X}$. Any temporal variable is defined with a time index, such as $X(h)$ for discrete time with $h = 0,1,2,\dots$, and $X(t)$ for continuous time with $t\geq 0$. For any one-dimensional matrix $\bm{X}$, the $i$-th element is denoted by $[\bm{X}]_{i}$, and for any two-dimensional matrix $\bm{Y}$, the $(i,j)$-th element is denoted by $[\bm{Y}]_{ij}$, with both indexes $i$ and $j$ being zero-based, i.e., $i,j\geq0$.

\section{Markov processes}
\label{sec:markov_processes}
This section establishes the mathematical framework for analyzing both single and multiple independent Markov processes. We begin with the transient and asymptotic characterization of a single discrete-time Markov process, which, while seemingly elementary, provides the indispensable foundation for the more general case of aggregated multi-process systems. Building on this foundation, we extend our analysis to multiple independent Markov processes, yielding exact closed-form expressions for transient probabilities, conditional expectations, and asymptotic distributions. The results and lemmas presented in this section reveal structural properties that form the basis for devising scalable decision-making strategies in restless multi-armed bandits.

\subsection{Single Markov Process}
\label{sec:singleMP}
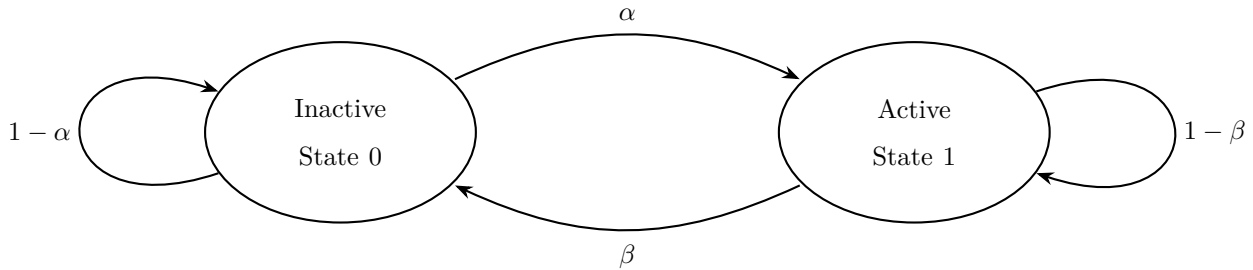
\begin{figure*}
    \centering
    \begin{tikzpicture}[
  >=Stealth, thick, font=\sffamily,
  node distance=40mm,           
  state/.style={draw,ellipse,minimum width=36mm,minimum height=24mm,align=center},
  xscale=.7
]
  \node[state] (S0) {Inactive\\State $0$};
  \node[state,right=of S0] (S1) {Active\\State $1$};

  \path[->] (S0) edge[loop left, min distance=10mm, looseness=10]
        node[left] {$1-\alpha$} (S0);
  \path[->] (S1) edge[loop right, min distance=10mm, looseness=10]
        node[right] {$1-\beta$} (S1);

  \draw[->,bend left=18,shorten >=1pt,shorten <=1pt]
        (S0) to node[above,yshift=2pt] {$\alpha$} (S1);
  \draw[->,bend left=18,shorten >=1pt,shorten <=1pt]
        (S1) to node[below,yshift=-2pt] {$\beta$} (S0);

\end{tikzpicture}
    \caption{State transition diagram of a discrete-time Markov process. The system alternates between an inactive state $0$ and an active state $1$ with activation probability $\alpha$ and deactivation probability $\beta$. Self-loops $(1-\alpha)$ and $(1-\beta)$ represent the probabilities of remaining in the same state at each time step.}
    \label{fig:single_MP}
\end{figure*}
Consider a discrete-time, ergodic, and time-homogeneous Markov process defined on a finite state space $\mathcal{S}=\{0,1\}$. At any given discrete time step $h\geq0$, the state of the Markov chain is represented by a binary random variable $X(h)$, with $X(h)=0$ indicating the inactive state, and $X(h)=1$ indicating the active state, as shown in Figure \ref{fig:single_MP}. Transitions between states are governed by $\alpha, \beta\in(0,1)$, where $\alpha$ represents the probability of transitioning from the inactive state to the active state, and $\beta$ represents the probability of transitioning from the active state to the inactive state. Together, these transition probabilities fully characterize the stochastic evolution of the process, and they are captured in the left-stochastic transition probability matrix
\begin{equation}
    \bm{P} = \begin{bmatrix}
    1-\alpha & \beta  \\
    \alpha & 1-\beta
    \end{bmatrix}.
    \label{eq:transProbMatrix}
\end{equation}

Given $\bm{P}$, the state probability vector $\bar{\bm{X}}(h) = [\mathcal{P}\big(X(h)=0\big),\mathcal{P}\big(X(h)=1\big)]^\mathsf{T}$ evolves as $\bar{\bm{X}}(h+1) = \bm{P}\,\bar{\bm{X}}(h)$, $\forall h\geq 0$. By iterating this recurrence,  the probability distribution at $m$ steps ahead of time $h$ can be expressed as
\begin{equation}
    \bar{\bm{X}}(h+m) = \bm{P}^m\,\bar{\bm{X}}(h), \qquad \forall h,m\geq 0.
    \label{eq:recurrenceRelation}
\end{equation}
If the initial distribution $\bar{\bm{X}}(0)$ is known, the probability distribution evolves as $\bar{\bm{X}}(m) = \bm{P}^m\,\bar{\bm{X}}(0)$, $\forall m\geq0$. Alternatively, an observation of the state $X(h)$ at time $h$ yields a degenerate distribution, encoded as a one-hot vector $\bm{X}(h)=[1-X(h),X(h)]^\mathsf{T}$, and the probability distribution at any future time step is $\bar{\bm{X}}(h+m) = \bm{P}^m\bm{X}(h)$, $\forall m\geq0$. 

Since $\bm{P}$ is diagonalizable, we can use the eigenvalue decomposition $\bm{P}=\bm{V}\bm{\Lambda}\bm{V}^{-1}$, where
\begin{equation}
    \bm{\Lambda} = \begin{bmatrix}
    1 & 0 \\
    0 & 1-(\alpha+\beta)
    \end{bmatrix}, \qquad \bm{V} = \begin{bmatrix}
    \beta & -1  \\
    \alpha & 1
    \end{bmatrix}, \qquad \bm{V}^{-1} = \frac{1}{\alpha+\beta}\begin{bmatrix}
    1 & 1 \\
    -\alpha & \beta
    \end{bmatrix},
\end{equation}
with the inverse matrix well-defined since $\alpha + \beta > 0$. Using this decomposition, the $m$-step transition probability matrix is
\begin{equation}
    \bm{P}^m = \bm{V}\bm{\Lambda}^m\bm{V}^{-1}
    = \frac{1}{\alpha+\beta}\begin{bmatrix}
    \beta + \alpha (1-\alpha-\beta)^m  & \beta - \beta (1-\alpha-\beta)^m \\
    \alpha- \alpha (1-\alpha-\beta)^m  & \alpha + \beta (1-\alpha-\beta)^m 
    \end{bmatrix},
    \label{eq:transProbMatrix_multipleSteps}
\end{equation}
where
\begin{equation*}
    [\bm{P}^m]_{ij} = \mathcal{P}\big(X(h+m)=i\,|\,X(h)=j\big), \qquad \forall i,j\in\mathcal{S}, \quad \forall m\geq 0, \quad \forall h\geq 0.
\end{equation*}
The closed-form expression in \eqref{eq:transProbMatrix_multipleSteps} not only facilitates the efficient computation of multi-step transitions, but it also reveals the asymptotic behavior of the process. Since $\alpha$ and $\beta$ are defined in the range $(0,1)$, the second eigenvalue of $\bm{P}$ always satisfies $|1-\alpha-\beta|<1$, and the term $(1 - \alpha - \beta)^m\to 0$ as $m\to\infty$. Consequently, in the limit, the state probability vector converges to the stationary distribution $\bm{\pi}$, defined by
\begin{equation}
    \bm{\pi} = \lim_{m\rightarrow\infty} \bar{\bm{X}}(h+m) = \frac{1}{\alpha+\beta}\begin{bmatrix}
    \beta  & \beta \\
    \alpha & \alpha
    \end{bmatrix}
    \begin{bmatrix}
    \mathcal{P}\big(X(h)=0\big) \\
    \mathcal{P}\big(X(h)=1\big)
    \end{bmatrix} = \frac{1}{\alpha+\beta}\begin{bmatrix}
    \beta\\
    \alpha
    \end{bmatrix},
    \label{eq:stationayDist_singleChain}
\end{equation}
regardless of the initial state, as $\mathcal{P}\big(X(h)=0\big) +\mathcal{P}\big(X(h)=1\big) =1$.

With that, we have characterized the behavior of a single Markov process, including its transient dynamics and stationary distribution, and we are now ready to extend this framework to multiple Markov processes.


\subsection{Multiple Markov processes}
\label{sec:multiMP}
Consider $N$ independent, discrete-time, ergodic, and time-homogeneous Markov processes, each with transition probability matrix $\bm{P}$, as defined in \eqref{eq:transProbMatrix}. At any discrete time $h \geq 0$, assume we do not have access to the individual states of these processes, but instead, we observe the aggregate number of active processes, represented by the random variable $Y(h)\in\{0,1,\dots,N\}$. Our goal is to characterize the conditional probability distribution of $Y(h+m)$ over time and explore its long-term behavior, leveraging the mathematical framework established in the previous section.

Since each process evolves independently, the aggregate behavior of the system can be analyzed by considering the transitions of active and inactive processes separately. Given $Y(h) = j$, the number of active processes that remain active after $m$ steps is $\text{Binomial}(j, q_m)$ with
\begin{equation}
    q_m=[\bm{P}^m]_{11} = \frac{1}{\alpha+\beta}\Big[ \alpha + \beta(1-\alpha-\beta)^m \Big],
    \label{eq:transition_AA}
\end{equation}
and the number of inactive processes that become active after $m$ steps is $\text{Binomial}(N - j, p_m)$ with
\begin{equation}
    p_m=[\bm{P}^m]_{10} = \frac{1}{\alpha+\beta}\Big[ \alpha- \alpha (1-\alpha-\beta)^m \Big].
    \label{eq:transition_IA}
\end{equation}
Taken together, the conditional probability of $Y(h+m)$ given $Y(h)=j$ can be expressed as
\begin{equation}
\mathcal{P}\big(Y(h+m) = y \mid Y(h) = j\big) = \sum_{\ell = \max(0,y+j-N)}^{\min(y,j)} \binom{j}{\ell} q_m^{j-\ell} (1-q_m)^\ell \, \binom{N-j}{y-\ell} p_m^{y-\ell} (1-p_m)^{N+\ell-y-j},
\label{eq:transProbMatrix_multipleProcesses}
\end{equation}
for any integers $m\geq0$, $h\geq0$, and $y,j\in\{0,1,\dots,N\}$, where $q_m$ and $p_m$ are defined in \eqref{eq:transition_AA} and \eqref{eq:transition_IA}, respectively. The summation bounds $\max(0,y+j-N)$ and $\min(y,j)$ ensure that both binomial coefficients remain valid.

Building upon the closed-form expressions of the transient probabilities in equation \eqref{eq:transProbMatrix_multipleProcesses}, we can derive additional statistics that provide deeper insights into the short-term and long-term behavior of the system. One important measure is the conditional expectation of the number of active processes. Since the total number of active processes is the sum of two independent Binomial random variables, one representing the number of initially active processes that remain active, and the other representing the number of initially inactive processes that become active, we can use the linearity of expectation to find that
\begin{equation}
    \mathbb{E}\big[Y(h+m)  \mid Y(h) = j\big] = N \frac{\alpha}{\alpha+\beta} + \left(1-\alpha-\beta\right)^m \left(j-N \frac{\alpha}{\alpha+\beta}\right).
    \label{eq:expectedActive_multiProcess}
\end{equation}

Another important measure to form the basis of our decision-making framework is the rate of convergence of the entire system. Notice that the term $(1-\alpha-\beta)^m$ in both the expectation and the expressions for $p_m$ and $q_m$ dictates the decay of the transient effects. Since $|1-\alpha-\beta|<1$, the convergence is geometric, with the rate determined by $\alpha+\beta$. Large values of $\alpha+\beta$ lead to a faster decay of the transient term, while smaller values prolong the influence of the initial state $Y(h)=j$. In the asymptotic regime, this transient term vanishes, and both the expectation and the probability distribution converge to a well-defined limit. In particular, as the number of steps $m$ approaches infinity, the transition probabilities $q_m$ and $p_m$ converge to the steady-state probability of being active $\lim_{m\to\infty} q_m = \lim_{m\to\infty} p_m = \alpha/(\alpha+\beta)$, and the conditional probability in equation \eqref{eq:transProbMatrix_multipleProcesses} reduces to
\begin{equation}
    \lim_{m\to\infty} \mathcal{P}\big(Y(h+m) = y \mid Y(h) = j\big) = \binom{N}{y} \pi^y (1-\pi)^{N-y}, \qquad \text{with} \qquad \pi = [\bm{\pi}]_1=\frac{\alpha}{\alpha+\beta},
    \label{eq:asympProbMatrix_multipleProcesses}
\end{equation}
irrespective of the initial number of active processes $Y(h)=j$. This asymptotic behavior also aligns with the limiting value of the expected number of active processes,  $\lim_{m\to\infty} \mathbb{E}\big[Y(h+m) \mid Y(h) = j\big] = N\alpha/(\alpha+\beta)$, which corresponds to the mean of the limiting Binomial distribution presented in \eqref{eq:asympProbMatrix_multipleProcesses}.

In summary, the transient probabilities provide a detailed description of the system's evolution over a finite number of steps, whereas the asymptotic analysis confirms convergence to a well-defined stationary distribution. These transient and asymptotic results serve as a building block for the more intricate spatio-temporal analysis in the next section, where we extend these results to limiting regimes with infinite Markov processes and continuous time.

\subsection{Spatio-temporal analysis}
For an increasing number of processes $N$, calculating the transient behavior of independent Markov processes via equation~\eqref{eq:transProbMatrix_multipleProcesses} becomes computationally intensive, limiting its practical applicability in large-scale settings. To overcome this limitation, we derive asymptotic approximations for the conditional distribution of $Y(h+m)$ as $N$ approaches infinity. We specifically examine this problem in two settings, namely, a discrete-time model with vanishing activation probabilities, and a continuous-time model with infinitesimal transition rates, as illustrated schematically in Figure~\ref{fig:multiple_MP}.
\begin{figure*}
    \centering
    \begin{tikzpicture}[
  >=Stealth, thick, font=\sffamily,
  x=1mm,y=1mm,
  state0blue/.style={draw=blue!15!white,fill=blue!20,circle,minimum size=9.5mm,align=center},
  state0neutral/.style={draw=black!70!white,circle,minimum size=9.5mm,align=center},
  state1neutral/.style={draw=black!70!white,circle,minimum size=9.5mm,align=center},
  state1active/.style={draw=red!20!white,fill=red!25,circle,minimum size=9.5mm,align=center}
]

\def\nRows{3}
\def\nCols{3}
\def\dx{24}        
\def\rowdy{15}     
\def\colsep{55}    
\def\bracegap{13}  
\def\labeldrop{8} 
\def\boxdrop{3} 
\def\boxshift{44}  

\newcommand{\redcirc}{\tikz[baseline=-0.6ex]{\filldraw[draw=red!20!white,fill=red!25] (0,0) circle (4.3pt);}}

\foreach \r in {1,...,\nRows} {
  \pgfmathsetmacro{\y}{-(\r-1)*\rowdy}
  \foreach \c in {1,...,\nCols} {
    \pgfmathsetmacro{\x}{(\c-1)*\colsep}
    \pgfmathtruncatemacro{\isActive}{int((\c==1 && \r==1) || (\c==3 && \r==1) || (\c==2 && \r==2) || (\c==3 && \r==3))}
    \ifnum\isActive=1
      \node[state0neutral] (S0-\c-\r) at (\x,\y) {0};
      \node[state1active]  (S1-\c-\r) at (\x+\dx,\y) {1};
    \else
      \node[state0blue]    (S0-\c-\r) at (\x,\y) {0};
      \node[state1neutral] (S1-\c-\r) at (\x+\dx,\y) {1};
    \fi
    \path[->] (S0-\c-\r) edge[loop left,  min distance=7.5mm, looseness=10] (S0-\c-\r);
    \path[->] (S1-\c-\r) edge[loop right, min distance=7.5mm, looseness=10] (S1-\c-\r);
    \draw[->,bend left=22] (S0-\c-\r) to (S1-\c-\r);
    \draw[->,bend left=22] (S1-\c-\r) to (S0-\c-\r);
  }
}

\path let
  \p1 = (S0-1-\nRows.west),       
  \p2 = (S1-\nCols-\nRows.east),  
  \p3 = (S0-1-\nRows.south)       
in
  coordinate (braceL)   at (\x1-20, \y3-\bracegap)
  coordinate (braceR)   at (\x2+20, \y3-\bracegap)
  coordinate (braceMid) at ($(braceL)!0.5!(braceR)$);

\draw[decorate, decoration={mirror, brace, amplitude=8pt,raise=1.5pt}]
  (braceL) -- (braceR);

\node[anchor=north,align=center] (ylabel) at ($(braceMid)+(0,-\labeldrop)$) {%
  $\displaystyle Y(h)\;=\;\sum\;\redcirc$\\[5pt]
  \scriptsize number of active processes
};

\node[anchor=north, draw, rounded corners=2pt, fill=white,
      align=center, font=\normalsize, inner sep=6pt]
  at ($(ylabel.north)+(-\boxshift,\boxdrop)$) {%
  \textbf{Lemma \ref{lemma:convergence_nrof_chains}}\\
  $\alpha=\bar{\alpha}/N,\ \bar{\alpha}>0$\\
  $\beta\in(0,1)$\\
  $N\rightarrow\infty$
};

\node[anchor=north, draw, rounded corners=2pt, fill=white,
      align=center, font=\normalsize, inner sep=6pt]
  at ($(ylabel.north)+(\boxshift,\boxdrop)$) {%
  \textbf{Lemma \ref{lemma:convergence_nrof_chains_CTMC}}\\
  $\alpha=\bar{\alpha}\Delta_t/N,\ \bar{\alpha}>0$\\
  $\beta=\bar{\beta}\Delta_t,\ \bar{\beta}>0$\\
  $N\to\infty,\ \Delta_t\to0$
};

\end{tikzpicture}
    \caption{Schematic representation of the ensemble of binary-state Markov processes considered in Lemmas~\ref{lemma:convergence_nrof_chains} and~\ref{lemma:convergence_nrof_chains_CTMC}. Each process alternates between an inactive state ($0$, blue) and an active state ($1$, red), with transition probabilities illustrated in Figure~\ref{fig:single_MP}. The right brace indicates the aggregate variable $Y(h)$, corresponding to the total number of active processes at time $h$. Lemma~\ref{lemma:convergence_nrof_chains} describes the limiting discrete-time case as $N \to \infty$, whereas Lemma~\ref{lemma:convergence_nrof_chains_CTMC} extends the result to the continuous-time limit as $N \to \infty$ and $\Delta_t \to 0$.}
    \label{fig:multiple_MP}
\end{figure*}
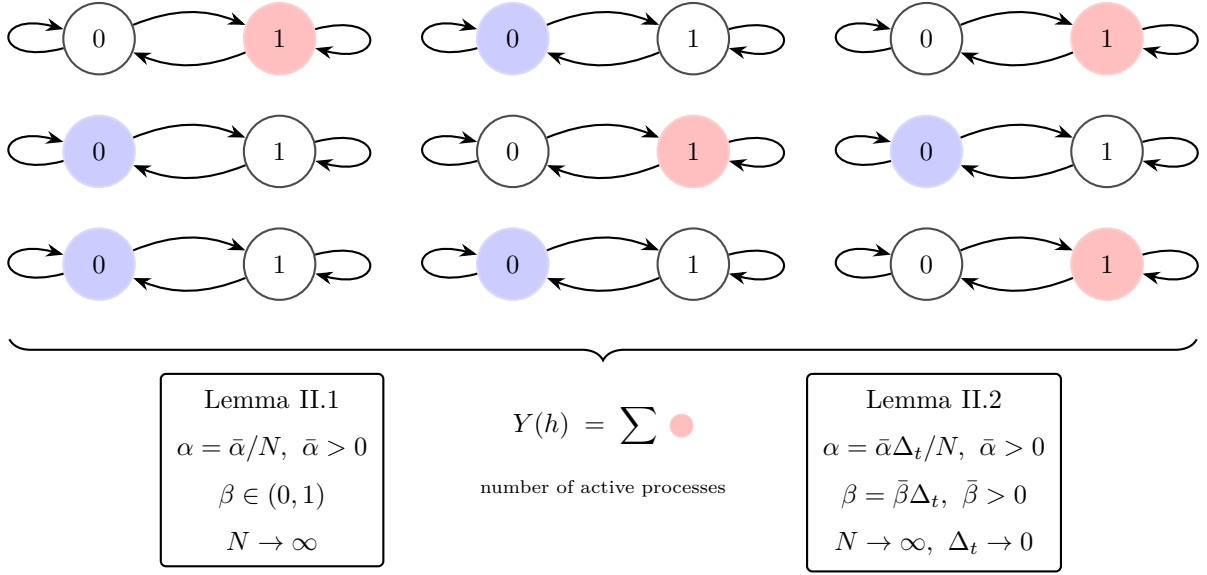

\begin{lemma}[Statistical properties of infinitely many independent discrete-time Markov processes]
\label{lemma:convergence_nrof_chains}
Consider $N$ independent, discrete-time Markov processes, each defined on the state space $\mathcal{S}\in\{0,1\}$ with transition probability matrix
\begin{equation*}
    \bm{P} = \begin{bmatrix}
    1-\alpha & \beta  \\
    \alpha & 1-\beta
    \end{bmatrix}.
\end{equation*}
Define the scaling $\alpha=\bar{\alpha}/N$, with $\bar{\alpha}>0$ and assume $\beta$ to be a fixed constant in $(0,1)$. Let $Y(h)$ denote the total number of processes in state $1$ at time $h$, taking values in $\{0,1,\dots,N\}$. For any integers $m\geq0$ and $j\geq0$, as $N\to\infty$, the conditional distribution of $Y(h+m)$ given $Y(h)=j$ converges in distribution to that of a random variable 
$Z+W$ where
\begin{itemize}
    \item $Z \sim \text{Binomial}\big(j,(1-\beta)^m\big)$,
    \item $W \sim \text{Poisson}\left(\frac{\bar{\alpha}}{\beta}\left(1-(1-\beta)^m\right)\right)$,
    \item $Z$ and $W$ are independent.
\end{itemize}

Moreover, its expected value approaches
\begin{equation}
    \mathbb{E}\left[Y(h+m) \mid Y(h) = j\right] \xrightarrow{N\to\infty} \frac{\bar{\alpha}}{\beta}\Big(1-(1-\beta)^m\Big) + j(1-\beta)^m.
    \label{eq:expectedActive_multiProcess_limitDiscrete}
\end{equation}
\end{lemma}
\begin{proof}
    See Appendix \ref{app:proof-TranProb_infiniteProcesses}.
\end{proof}


\begin{lemma}[Statistical properties of infinitely many independent continuous-time Markov processes]
\label{lemma:convergence_nrof_chains_CTMC}
Consider $N$ independent, discrete-time Markov processes, each defined on the state space $\mathcal{S}\in\{0,1\}$ with transition probability matrix
\begin{equation*}
    \bm{P} = \begin{bmatrix}
    1-\alpha & \beta  \\
    \alpha & 1-\beta
    \end{bmatrix}.
\end{equation*}
Define the scaling $\alpha=\bar{\alpha}\Delta_t/N$ and $\beta=\bar{\beta}\Delta_t$ over an infinitessimal time period $\Delta_t$, and assume $\bar{\alpha}$ and $\bar{\beta}$ to be fixed positive constants. Let $Y(t)$ denote the total number of processes in state $1$ at time $t$. For any time $\tau \geq 0$ and integer $j\geq0$, as $N\to\infty$ and $\Delta_t\to 0$, the conditional distribution of $Y(t+\tau)$ given $Y(t)=j$ converges in distribution to that of a random variable 
$Z+W$ where
\begin{itemize}
    \item $Z \sim \text{Binomial}\left(j,e^{-\bar{\beta}\tau}\right)$,
    \item $W \sim \text{Poisson}\left(\frac{\bar{\alpha}}{\beta}\left(1-e^{-\bar{\beta}\tau}\right)\right)$,
    \item $Z$ and $W$ are independent.
\end{itemize}

Moreover, its expected value approaches
\begin{equation}
    \mathbb{E}\left[Y(t+\tau) \mid Y(t) = j\right] \xrightarrow{N\to\infty, \; \Delta_t\to 0} \frac{\bar{\alpha}}{\bar{\beta}}\Big(1-e^{-\bar{\beta}\tau}\Big) + je^{-\bar{\beta}\tau}.
    \label{eq:expectedActive_multiProcess_limitContinous}
\end{equation}
\end{lemma}
\begin{proof}
    See Appendix \ref{app:proof-TranProb_CTMC}.
\end{proof}

The results from Lemmas \ref{lemma:convergence_nrof_chains} and \ref{lemma:convergence_nrof_chains_CTMC} provide a unified analytical framework for modeling the transient behavior of large ensembles of independent Markov processes in both discrete- and continuous-time settings. In the discrete-time case, Lemma~\ref{lemma:convergence_nrof_chains} quantifies how the number of active processes evolves over $m$ steps as $N$ becomes large, with the Binomial component reflecting the persistence of the initial state, and the Poisson component reflecting the accumulation of new active processes. Extending this to the continuous-time setting, Lemma~\ref{lemma:convergence_nrof_chains_CTMC} adapts the analysis to dynamics governed by differential transition rates, yielding a limiting distribution that combines an exponential memory decay for initially active processes with Poissonian arrivals of new activations, matching the intuition from the discrete-time result but adapted to a continuous-time framework. Taken together, these limiting distributions enable compact representations of the joint dynamics across arms in RMAB, and they facilitate the design of effective and scalable decision-making, as explored next.

\section{Restless Multi-Armed Bandit Framework}
\label{sec:restless_MAB_framework}

Consider a RMAB problem with a decision-making entity facing a finite set of arms $\mathcal{K} := \{1, 2, \ldots, K\}$. Each arm $k \in \mathcal{K}$ consists of $N_k$ independent, time-homogeneous, and ergodic Markov processes, each evolving on a common binary state space $\mathcal{S} = \{0, 1\}$, where state $0$ represents the inactive state and state $1$ the active state. The dynamics of each process within arm $k$ are governed by the transition probability matrix $\bm{P}_k$ as in \eqref{eq:transProbMatrix} with (assumed known) parameters $\alpha_k, \beta_k\in(0,1)$.

Within this setup, the decision entity is allowed to select only one arm at a time based on a deterministic or random policy. In particular, let the arm selected at time $h$ be denoted $A(h)\in\mathcal{K}$. Upon selecting arm $A(h)=a$, the decision entity observes $Y_a(h)$, corresponding to the aggregate number of active processes within that arm. Formally,
\begin{equation}
    Y_k(h) = \sum_{n\in\mathcal{N}_k}[\bm{Z}_k(h)]_{1n}, \quad \forall k\in\mathcal{K}, \, \forall h\geq 0.
    \label{eq:regret}
\end{equation}
where $\mathcal{N}_k\in\{0,1,\dots, N_k\}$ is the set of processes in arm $k$, and $\bm{Z}_k(h)\in\{0,1\}^{|\mathcal{S}|\times N_k}$ is the latent state of arm $k$ at time $h$, whose column $n\in\mathcal{N}_k$ is a one-hot encoded vector describing the state of process $n$ within arm $k$. This construction ensures that $\bm{Z}_k(h)$ satisfies $\sum_{i\in\mathcal{S}} [\bm{Z}_k(h)]_{in} = 1$, $\forall n\in\mathcal{N}_k$, $ \forall k\in\mathcal{K}$, and $\forall h\geq 0$, indicating that each individual Markov process can only be at exactly one sate at any given time.

Considering this, the objective of the decision entity is to design a function policy $\phi:\mathcal{B}\to \mathcal{K}$ that maps the belief states $\mathcal{B}=\prod_{k=1}^K \{0,\dots,N_k\}\times \mathbb{N}_0$ to arm selections and maximizes the average discounted reward over an infinite time horizon
\begin{equation}
    \mathcal{J}_\gamma(\phi) = \sum_{h=0}^{\infty}\gamma^{h}\; \mathbb{E}_\phi\left[Y_{A(h)}(h)\right], \qquad 0<\gamma<1,
\end{equation}
where $\gamma$ represents the discount factor reflecting the present value of future rewards, and where the expectation is taken with respect to the initial state
distribution and the joint distribution induced on all system variables.

One of the key challenges and a distinguishing feature of the RMAB setting is that all processes evolve regardless of selection. That is, the state transitions of all chains within arm $k$ occur according to their transition probability matrices $\bm{P}_k$, regardless of whether an arm is actively selected by the decision entity at time $h$. This partial observability and restless evolution only allows the decision entity to make decisions relying on past observations and known transition dynamics. In that regard, classical RMAB settings let the decision entity construct a belief state $\bar{\bm{Z}}_k(h)$ for all arms as
\begin{equation}
    \bar{\bm{Z}}_k(h) = \begin{cases}
        \bm{Z}_k(h) & A(h) = k, \\
        \bm{P}_k \, \bar{\bm{Z}}_k(h-1) & A(h)\neq k,
    \end{cases}
    \label{eq:belief-update-matrix}
\end{equation}
which exploits the independence of the chains and the Markov property by propagating the belief forward with $\bm{P}_k$ whenever arm $k$ is not selected. However, maintaining the $|\mathcal{S}|\times N_k$ belief matrices for all $k\in\mathcal{K}$ is computationally prohibitive, motivating the search for a compact representation of arm states that leverages the full stochastic structure of the problem. 

Building upon the analytical results from Section~\ref{sec:markov_processes}, we identify the tuple $(j_k, m_k)_{k\in\mathcal{K}}$ as a minimal sufficient statistic for arm $k$, where $j_k \in \{0, 1, \ldots, N_k\}$ is the number of active processes when arm $k$ was last observed, and $m_k \in\mathbb{N}_0$ is the delay since that last observation. These minimal sufficient statistics summarize the arm states, allowing the global state space to reduce from exponential to linear complexity in $K$. In such a setting, the belief state for each arm $k \in \mathcal{K}$ is updated from its current value $(j_k, m_k)$ to the next value $(j_k', m_k')$ as follows
\begin{equation}
    (j_k',\, m_k') \leftarrow \begin{cases}
        (Y_k(h),\, 0) & A(h) = k, \\
        (j_k,\, m_k+1) & A(h)\neq k.
    \end{cases}
    \label{eq:belief-update}
\end{equation}

Adopting $(j_k,m_k)_{k\in\mathcal{K}}$ as a sufficient statistic addresses the partial observability and restless nature of the RMAB problem, yet the pursuit of an optimal policy through exact dynamic programming remains computationally prohibitive. Indeed, the RMAB problem is known to be PSPACE-hard, even when the state matrix $\bm{Z}(h)$ is fully observable~\cite{papadimitriou1999complexity}. This computational intractability has motivated the development of heuristic approaches that balance optimality and scalability~\cite{whittle1988restless, weber1990index}. Among the various heuristic strategies proposed in the literature, index policies have emerged as a particularly promising approach due to their ability to decompose the complex multi-armed optimization into a collection of single-armed subproblems. The fundamental insight underlying index policies is that each arm can be assigned a scalar index that captures its relative priority for being selected, thereby reducing the combinatorial selection problem to a simple ranking and selection procedure. The Whittle index, introduced in~\cite{whittle1988restless}, represents the most principled instantiation of this approach, and there is strong evidence to suggest that it provides near-optimal performance with theoretical guarantees in many settings~\cite{glazebrook2006some,akbarzadeh2022conditions,wang2019whittle}.


The Whittle index is derived through a Lagrangian relaxation of the original optimization problem. Instead of enforcing that exactly one arm must be selected at each time, we consider a relaxed problem where each arm receives a subsidy $\lambda$ for remaining unobserved. This relaxation transforms the constrained optimization problem into $K$ decoupled single-arm Markov decision problems, each seeking to balance the immediate reward of being selected with the subsidy costs of remaining unobserved~\cite{bertsimas2000restless}. Importantly, as established in~\cite{whittle1988restless}, the Whittle index policy is meaningful only when all arms are indexable, a condition we verify for each arm $k \in \mathcal{K}$ in Appendix~\ref{app:proof-Whittle-Index}. Under these indexability conditions, the Whittle index, denoted by $\mathcal{W}(j_k, m_k; \gamma)$, is well-defined and applicable in our RMAB formulation, and it represents the smallest subsidy $\lambda$ for which the decision entity is indifferent between selecting an arm or leaving it unobserved. Theorem \ref{theo:whittle_index} formalizes this derivation.


\begin{theorem}[Whittle index for independent Markov processes]
\label{theo:whittle_index}
Consider a restless multi-armed bandit problem with a finite set of arms $\mathcal{K}=\{1,2,\dots,K\}$, each consisting of $N_k$ independent, discrete-time, ergodic, and time-homogeneous Markov processes. Let $Y_k(h)$ denote an observable random variable at discrete time $h$, representing a statistic of interest associated with arm $k$. Let the belief state associated with arm $k\in\mathcal{K}$ be denoted by the tuple $(j_k,m_k)$ where $j_k$ is the most recent observation of arm $k$, and $m_k$ is the number of time steps since that observation. Suppose the conditional expectation of $Y_k(h)$, given that arm $k$ was last observed $m_k$ steps ago with value $j_k$, satisfies the affine form
\begin{equation}
    \mathbb{E}[Y_k(h) \mid Y_k(h-m_k)=j_k] = C(m_k) + j_k \, D(m_k),
\end{equation}
for some functions $C:\mathbb{N}_0\to\mathbb{R}$ and $D:\mathbb{N}_0\to\mathbb{R}$. Then, for any discount factor $\gamma\in(0,1)$, the Whittle index associated with arm $k$ in state $(j_k,m_k)$ is given by
\begin{equation}
    \mathcal{W}(j_k,m_k;\gamma) = \frac{1}{1-\gamma D(0)}\Big[ C(m_k) - \gamma C(m_k+1) +\gamma C(0)\Big] + \frac{j_k}{1-\gamma D(0)}\Big[D(m_k)-\gamma D(m_k+1)\Big].
    \label{eq:general_WhittleIndex}
\end{equation}
\begin{proof}
    See Appendix \ref{app:proof-Whittle-Index}.
\end{proof}

\end{theorem}

\begin{corollary}[Discrete-time Whittle index for binary-state Markov processes] 
\label{cor:whittle_discrete}
Let each arm $k \in\mathcal{K}= \{1, \dots, K\}$ consist of $N_k$ independent, discrete-time, ergodic Markov processes evolving on a binary state space with transition probability matrix
\begin{equation*}
    \bm{P}_k = \begin{bmatrix}
    1-\alpha_k & \beta_k  \\
    \alpha_k & 1-\beta_k
    \end{bmatrix}
\end{equation*}
for some fixed constants $\alpha_k,\beta_k\in(0,1)$. Let $Y_k(h) \in \{0, \dots, N_k\}$ denote the number of processes in state $1$ at time $h$, let $j_k\in \{0, \dots, N_k\}$ denote the most recent observed value of $Y_k(h)$, and let $m_k\in\mathbb{N}_0$ denote the number of steps since that observation. Then, using the conditional expectation from equation~\eqref{eq:expectedActive_multiProcess} into \eqref{eq:general_WhittleIndex} yields that the Whittle index associated with arm $k$ in state $(j_k, m_k)$ is
\begin{equation}
    \mathcal{W}(j_k,m_k;\gamma) = N_k\frac{\alpha_k}{\alpha_k + \beta_k} + \frac{1-\gamma(1-\alpha_k-\beta_k)}{1-\gamma} (1-\alpha_k-\beta_k)^{m_k}\left(j_k - N_k\frac{\alpha_k}{\alpha_k + \beta_k}\right), \quad \forall k\in\mathcal{K}.
\end{equation}
\end{corollary}

\begin{corollary}[Discrete-time Whittle index with vanishing activation rates]
\label{cor:whittle_discrete_vanishing}
Under the same assumptions as in Corollary~\ref{cor:whittle_discrete}, consider further that the activation rate scales with the number of processes as $\alpha_k = \bar{\alpha}_k / N_k$, for a fixed constant $\bar{\alpha}_k > 0$, and let $N_k \to \infty$ for all $k\in\mathcal{K}$. Then, using the limiting conditional expectation from Lemma \ref{lemma:convergence_nrof_chains} into \eqref{eq:general_WhittleIndex} yields the Whittle index for arm $k$ in state $(j_k, m_k)$:
\begin{equation}
    \mathcal{W}(j_k,m_k;\gamma) = \frac{\bar{\alpha}_k}{\beta_k} + \frac{1-\gamma(1-\beta_k)}{1-\gamma} (1-\beta_k)^{m_k}\left(j_k - \frac{\bar{\alpha}_k}{\beta_k}\right), \quad \forall k\in\mathcal{K}.
\end{equation}
\end{corollary}

\begin{corollary}[Continouos-time Whittle index with vanishing activation rates]
\label{cor:whittle_continous_vanishing}
Let each arm $k\in\mathcal{K}= \{1, \dots, K\}$ consist of $N_k$ independent continous-time Markov processes evolving on a binary state space with transition rates $\bar{\alpha}_k/N_k$ from state $0$ to state $1$, and transition rate $\bar{\beta}_k$ from state $1$ to state $0$, for some fixed constants $\bar{\alpha}_k,\bar{\beta}_k>0$ and $N_k\to\infty$ for all $k\in\mathcal{K}$. Let $Y_k(t) \in \mathbb{N}_0$ denote the number of processes in state $1$ at time $t$, let $j_k \in \mathbb{N}_0$ denote the most recent observation of arm $k$, and let $\tau_k\in\mathbb{R}_0$ denote the time elapsed since that last observation. Then, using the limiting conditional expectation from Lemma \ref{lemma:convergence_nrof_chains_CTMC} into \eqref{eq:general_WhittleIndex} yields the Whittle index for arm $k$ in state $(j_k, \tau_k)$:
\begin{equation}
    \mathcal{W}(j_k,\tau_k;\gamma) = \frac{\bar{\alpha}_k}{\bar{\beta}_k} + \frac{1-\gamma e^{-\bar{\beta}_k}}{1-\gamma} e^{-\bar{\beta}_k \tau_k}\left(j_k - \frac{\bar{\alpha}_k}{\bar{\beta}_k}\right), \quad \forall k\in\mathcal{K}.
\end{equation}
\end{corollary}

As seen in Corollaries \ref{cor:whittle_discrete}--\ref{cor:whittle_continous_vanishing}, the analytical structure of the Whittle index consists of two parts. The first term reflects the long-run expected value of arm $k\in\mathcal{K}$, corresponding to the steady-state behavior of the underlying stochastic process; and the second term captures the transient deviation from steady state, weighted by the evolution of the system since the last observation and the discount factor $\gamma$. This form naturally balances exploitation of arms known to yield high long-term rewards against exploration of arms whose latent state may have been favorable when they were last observed.

With the Whittle indices defined for each arm as per Theorem~\ref{theo:whittle_index} and its corollaries, the proposed Whittle index policy establishes that, at each discrete-time decision epoch $h \geq 0$, the decision entity selects the arm $A(h) \in \mathcal{K}$ that maximizes the Whittle index given the current belief state $(j_k, m_k)_{k \in \mathcal{K}}$ and discount factor $\gamma \in (0,1)$,
\begin{equation}
    A(h) = \arg \max_{k \in \mathcal{K}} \;\mathcal{W}(j_k, m_k; \gamma),
    \label{eq:Whittle-index-policy}
\end{equation}
where ties are broken arbitrarily. The belief states are subsequently updated according to~\eqref{eq:belief-update}, ensuring that the policy adapts to the restless evolution of unobserved arms. Similarly, in the continuous-time setting, where decisions are made at arbitrary times $t \geq 0$ and the belief state for each arm $k$ is represented by $(j_k, \tau_k)_{k \in \mathcal{K}}$ with $\tau_k$ denoting the time elapsed since the last observation, the policy selects the arm $A(t) \in \mathcal{K}$ that maximizes the Whittle index:
\begin{equation}
    A(t) = \arg \max_{k \in \mathcal{K}} \;\mathcal{W}(j_k, \tau_k; \gamma),
\end{equation}
with ties broken arbitrarily. Upon selection of arm $A(t) = k$, the belief state updates to $(Y_k(t), 0)$, while for unselected arms, $\tau_k$ increments by the time until the next decision epoch. This continuous-time adaptation maintains the same decomposability and tractability benefits as in the discrete-time setting, with computational complexity $O(K)$ per decision.

\begin{corollary}[Myopic index policy]
In the limiting undiscounted case where $\gamma \to 0$, the Whittle index policy reduces to the myopic policy, which selects the arm with the highest expected reward,
\begin{align*}
    A(h) &= \arg \max_{k\in\mathcal{K}} \; \mathbb{E}\left[Y_k(h)\mid Y_k(h-m_k) = j_k\right], \qquad \text{(discrete time)},\\
    A(t) &= \arg \max_{k\in\mathcal{K}} \; \mathbb{E}\left[Y_k(t)\mid Y_k(t-\tau_k) = j_k\right], \qquad \text{(continuous time)}.
\end{align*}    
\end{corollary}

\begin{remark}[Extension to rested multi-arm bandits]
The rested multi-armed bandit problem is a special case of the RMAB problem in which passive arms remain frozen and do not evolve stochastically when unobserved. In that case, the Whittle index policy reduces to what is called the Gittins index policy and is optimal~\cite{gittins1979bandit}. Thus, the results obtained in this paper are also applicable to the rested multi-armed bandits.
\end{remark}

\begin{remark}[Extension to multiple arm selection]
While this paper focuses on the single-arm selection constraint, the proposed Whittle index policy naturally extends to the more general case where the decision entity may select $L\leq K$ arms at each time step. In this generalized setting, the optimal policy $\mathcal{A}(h)\subseteq \mathcal{K}$ under the Whittle index framework becomes:
\begin{equation}
    \mathcal{A}(h) = \arg\max_{\Theta \subseteq \mathcal{K}, |\Theta|=L} \sum_{k \in \Theta} \mathcal{W}(j_k,m_k;\gamma),
\end{equation}
which is equivalent to selecting the $L$ arms with the highest Whittle indices $\mathcal{A}(h) = \{k_1, k_2, \ldots, k_L\} $ such that
\begin{equation}
\mathcal{W}(j_{k_1}, m_{k_1}; \gamma) \geq \mathcal{W}(j_{k_2}, m_{k_2}; \gamma) \geq \cdots \geq \mathcal{W}(j_{k_L}, m_{k_L}; \gamma) \geq \max_{k \notin \mathcal{A}(h)} \mathcal{W}(j_k, m_k; \gamma),
\end{equation}
for any discount factor $\gamma\in(0,1)$. This decomposition property follows directly from the indexability of the relaxed problem and the additive structure of the Lagrangian relaxation. The belief state update rule in equation \eqref{eq:belief-update} remains unchanged, with the modification that $(j_k, m_k)$ is updated to $(Y_k(h), 0)$ if $k \in \mathcal{A}(h)$, and to $(j_k, m_k + 1)$ otherwise. The computational complexity becomes $\mathcal{O}(K \log K)$, required for sorting $K$ indices, making this extension practically tractable even for large-scale problems.
\end{remark}

\section{Simulation results}
\label{sec:simulation_results}
In this section, we evaluate the performance of the proposed Whittle index policy through extensive numerical simulations. These experiments are designed to validate the theoretical derivations presented in Section \ref{sec:restless_MAB_framework}, particularly focusing on the policy's ability to handle the restless dynamics and partial observability inherent in RMPMAB problems. We simulate scenarios with varying numbers of arms $K$ and processes per arm $N_k$, using transition parameters $\alpha_k$ and $\beta_k$ drawn from different distributions to reflect homogeneous and heterogeneous arm behaviors. The time horizon is set to $T=4000$ steps, with the initial states for all processes being drawn randomly from the stationary distribution $\bm{\pi}_k = [\beta_k / (\alpha_k + \beta_k), \alpha_k / (\alpha_k + \beta_k)]$ to ensure ergodicity from the start, unless stated otherwise. In each scenario, the decision entity selects exactly one arm at each discrete time step $h$, receiving as reward the aggregate number of active processes $Y_k(h)$. All simulations are implemented in Python using NumPy for efficient computation of Markov process evolutions, and all results are averaged over $100$ independent Monte Carlo trials to ensure statistical reliability. The code is proprietary and subject to intellectual property protections, but access to it may be discussed with the corresponding authors upon reasonable request.

\subsection{Performance analysis}
To quantify the performance of the decision policy under this partially observed and dynamically evolving system, we employ strong regret as the primary metric. Strong regret compares the cumulative reward of a policy against that of genie-aided optimal policy which selects the arm with the maximum instantaneous reward $\max_k Y_k(h)$ at each step,
\begin{equation}
    \mathcal{R}(T) = \sum_{h=1}^{T} \max_{k\in\mathcal{K}} Y_{k}(h) - \sum_{h=1}^{T} Y_{A(h)}(h).
\end{equation}
This metric is particularly suitable in the context of our problem, as it captures the full opportunity cost of suboptimal decisions in dynamic and partially observable environments and gives an accurate representation of the transient phenomena in the selected arms~\cite{jung2019regret}. Its counterpart, the weak regret, which compares policies against a simplified or partially informed benchmark, e.g., selecting the arm with the highest steady-state mean, underestimates costs in dynamic settings, and as such we do not employ it to measure the performance of our proposed index policy.

Considering the above, the simulation framework initializes the decision policy by allowing it to observe each arm exactly once in the initial $K$ time steps, ensuring that the policy acquires a baseline understanding of the system state. Thereafter, the policy begins making arm selections $A(h)$ for $h \geq K$ according to \eqref{eq:Whittle-index-policy}, leveraging the observed data to inform its decision-making process based on the Whittle index strategy. 

For comparison, we benchmark the Whittle index policy against established baselines, including
\begin{itemize}
    \item {
    \textit{Round Robin}: This deterministic policy cycles through the arms in a fixed order: $A(h) = (h \mod K) + 1$. It ignores state information and serves as a simple, non-adaptive baseline.
    }
    \item {
    \textit{$\epsilon$-Greedy}: With probability $\epsilon = .15$, the policy selects an arm uniformly at random (exploration); otherwise, it selects $A(h) = \arg\max_k \mu_k$ (explotation), where $\mu_k$ is the sample mean of arm $k$ up to time $h-1$.
    }
    \item {
    \textit{Bayesian Upper Confidence Bound}: Assuming sub-Gaussian rewards with parameter $c=1.96$, this policy selects arms according to $A(h) = \arg\max_k \left[ \mu_k + c \sqrt{\sigma_k / \mathcal{T}_k} \right]$, where $\mu_k$ is the sample mean, $\sigma_k$ is the sample variance, and $\mathcal{T}_k$ is the number of times arm $k$ has been selected up to time $h-1$.
    }
    \item {
    \textit{Thompson Sampling}: This probabilistic policy maintains the belief state $\bar{\bm{Z}}_k(h)$ for the distribution of the latent state via the update in \eqref{eq:belief-update-matrix}, then it samples the state of the processes from the posteriors and selects the arm that maximizes the estimated number of active chains.
    }
\end{itemize}

\subsection{Empirical evaluation}
Figure \ref{fig:regretHeter_time} illustrates the average cumulative strong regret over time for a heterogeneous arm configuration with $K=30$ arms, each comprising $N_k=100$ independent Markov processes. For each Monte Carlo trial, the transition parameters are sampled independently per arm as $\alpha_k \sim \mathcal{U}(0,1)$ and $\beta_k \sim \mathcal{U}(\alpha_k,1)$, ensuring arm heterogeneity while maintaining $\beta_k \geq \alpha_k$ to prevent outlier arms with disproportionately high steady-state rewards. This setup fosters a competitive environment where multiple arms exhibit comparable expected rewards, compelling policies to balance exploration and exploitation effectively rather than converging prematurely to a single dominant arm. As shown in the figure, the myopic Whittle index policy consistently achieves the lowest regret, outperforming the baselines by leveraging belief updates on latent process states to adapt to restless dynamics. In contrast, Round Robin incurs the highest regret due to its non-adaptive nature, while $\epsilon$-Greedy, Bayesian UCB, and Thompson Sampling exhibit intermediate performance. At $T=4000$, the Whittle index policy incurs a $37\%$ reduction in cumulative regret compared to Thompson Sampling, a $42\%$ reduction relative to Bayesian UCB, and a $52\%$ reduction compared to $\epsilon$-Greedy. The suboptimal performance of these standard algorithms arises from their reliance on empirical reward statistics, often without fully accounting for the restless nature of the arms where states evolve even when unobserved. Specifically, Round Robin lacks any adaptive mechanism or modeling of dynamics, resulting in rigidly cyclic selections that ignore latent state changes entirely; $\epsilon$-Greedy and Bayesian UCB focus on empirical means with random exploration or confidence bounds, respectively, but treat rewards as largely stationary, leading to inefficient belief updates and an inability to propagate transient dependencies for unselected arms. Even Thompson Sampling, which does incorporate belief propagation on latent states through Bayesian posteriors, suffers from sampling-induced randomness that can yield less accurate exploitation of evolving dynamics compared to the deterministic, index-driven prioritization of the Whittle policy.
\begin{figure*}[!t]
    \centering
    \subfigure[Heterogeneous arms]{\includegraphics[width=0.43\textwidth]{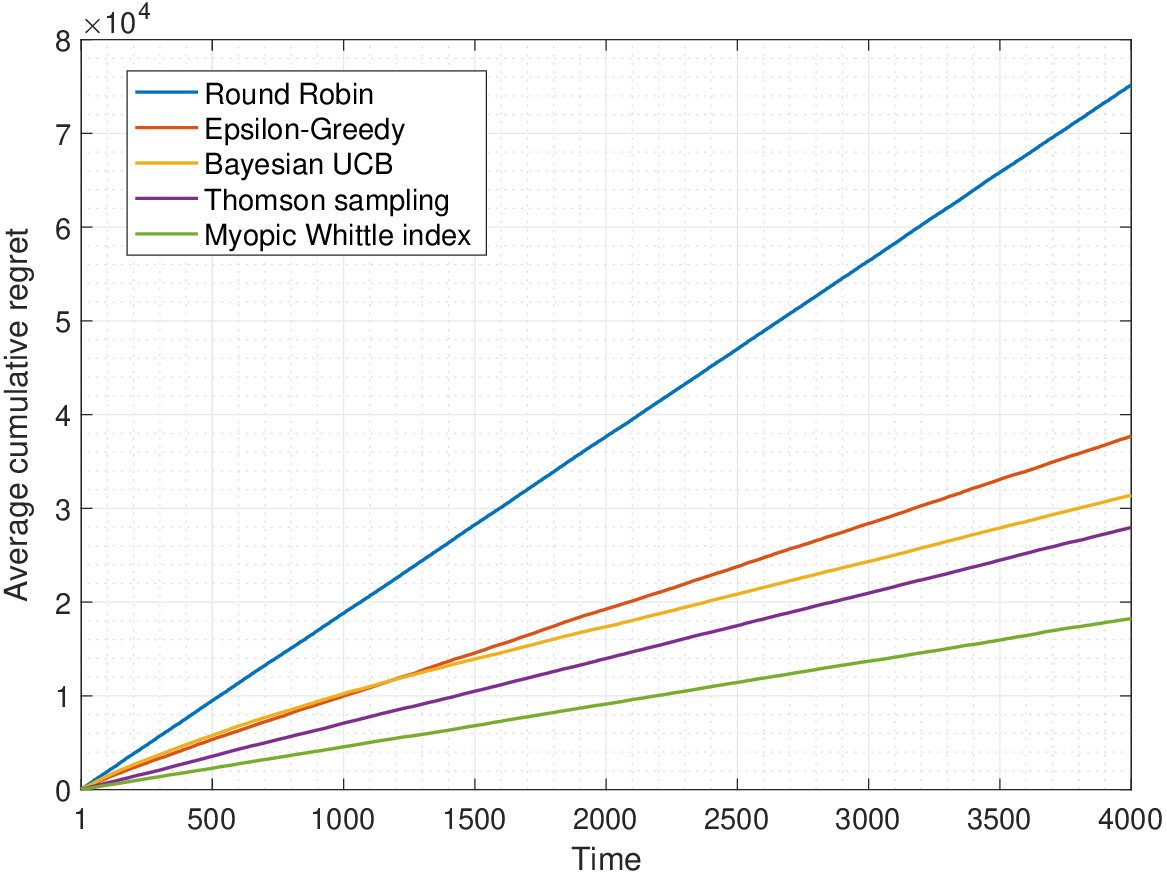} \label{fig:regretHeter_time}}$\quad$
    \subfigure[Homogeneous arms]{\includegraphics[width=0.43\textwidth]{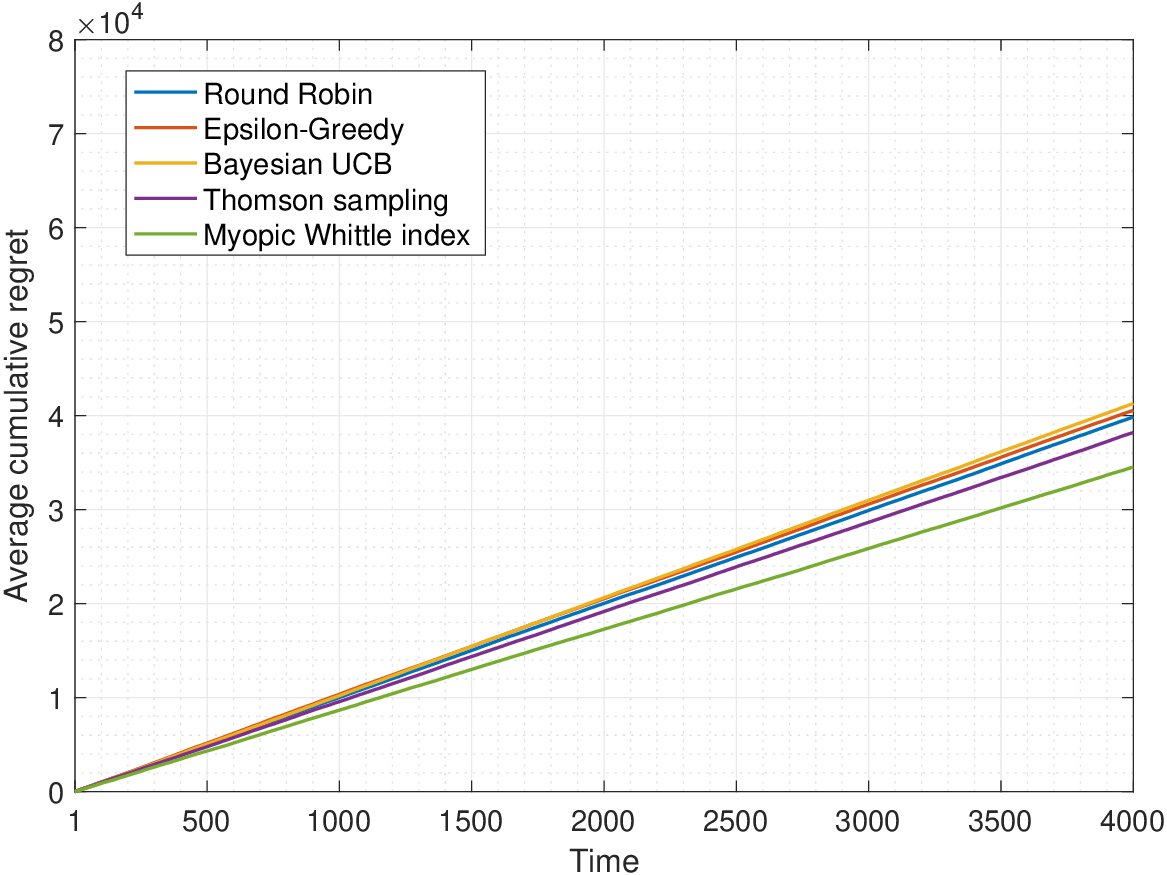} \label{fig:regretHomog_time}}
    \subfigure[Discrete time with vanishing activation rates]{\includegraphics[width=0.43\textwidth]{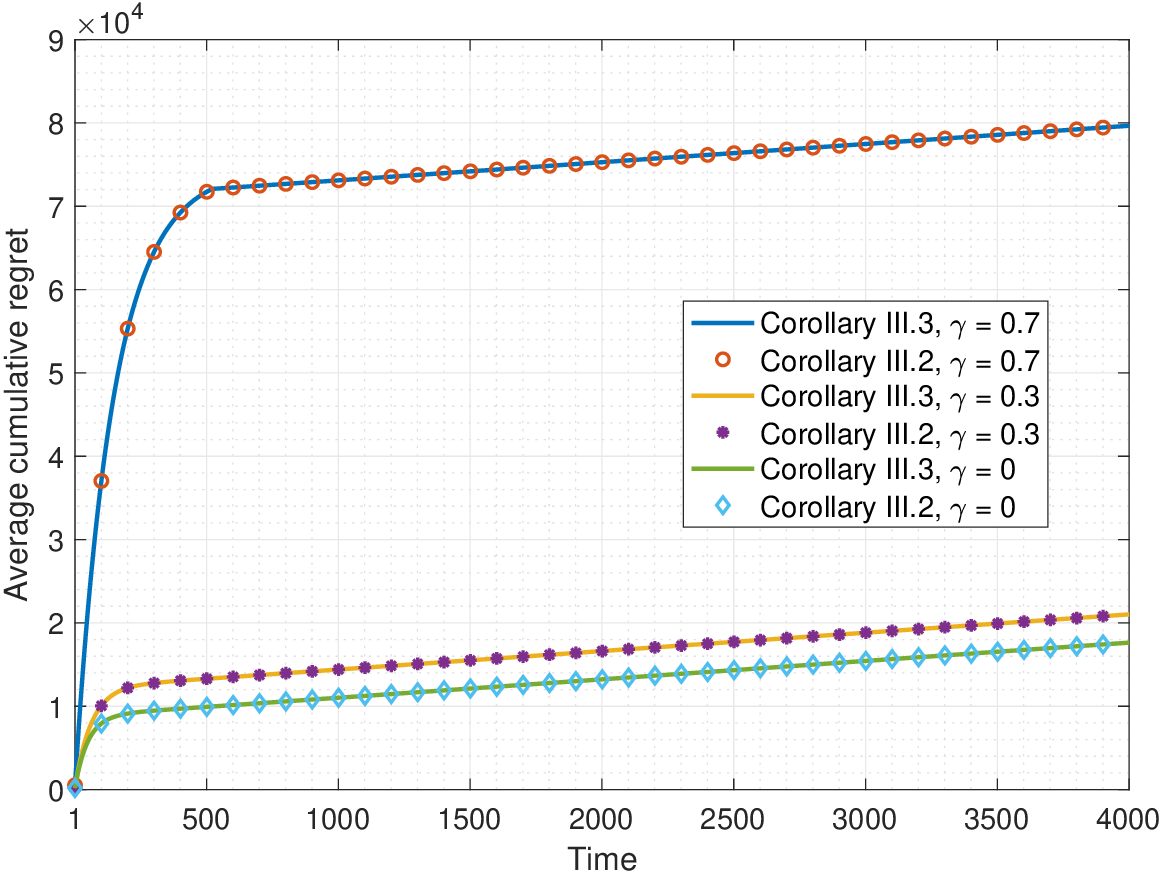} \label{fig:regretHeter_corollary3}}$\quad$
    \subfigure[Continous time with vanishing activation rates]{\includegraphics[width=0.43\textwidth]{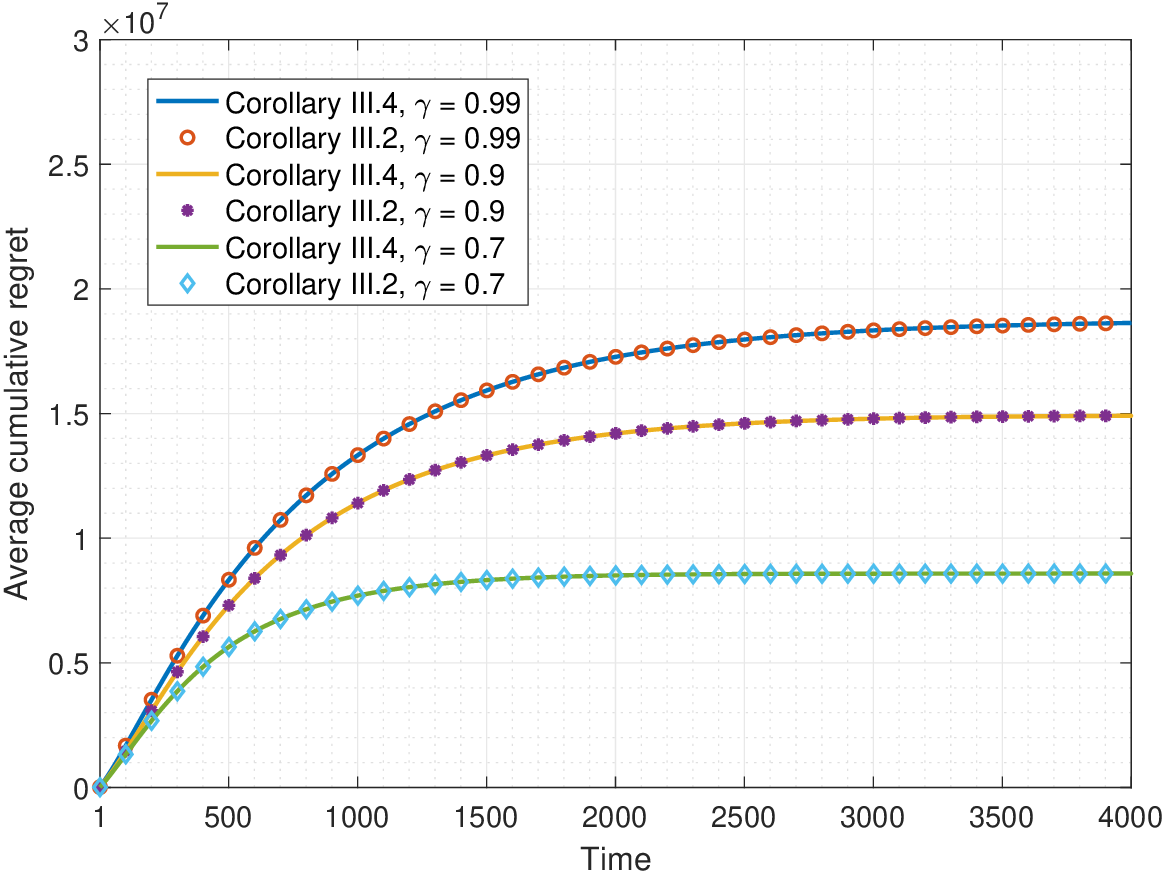} \label{fig:regretHeter_corollary4}}
    \caption{Average cumulative strong regret over time for the proposed restless multi-process multi-armed bandit (RMPMAB) framework. Figure~\ref{fig:regretHeter_time} shows results for a heterogeneous arm configuration, where the myopic Whittle index policy achieves the lowest regret and outperforms Round Robin, $\epsilon$-Greedy, Bayesian UCB, and Thompson Sampling. Figure~\ref{fig:regretHomog_time} presents results for a homogeneous arm configuration with identical transition parameters, where the performance gap narrows but the Whittle index remains superior. Figure~\ref{fig:regretHeter_corollary3} evaluates the discrete-time Whittle index policies (Corollaries~\ref{cor:whittle_discrete} and~\ref{cor:whittle_discrete_vanishing}) under vanishing activation rates and different discount factors, showing that both formulations yield consistent regret profiles, confirming theoretical equivalence. Figure~\ref{fig:regretHeter_corollary4} compares the continuous- and discrete-time Whittle index policies (Corollaries~\ref{cor:whittle_discrete} and~\ref{cor:whittle_continous_vanishing}), illustrating matching asymptotic behavior and the effect of the discount factor on convergence and regret growth.}
\end{figure*}

Figure \ref{fig:regretHomog_time} depicts the average cumulative strong regret over time for a homogeneous arm configuration comprising $K=30$ arms, each with $N_k=100$ independent Markov processes. The transition parameters are shared across all arms and sampled at the beginning of each Monte Carlo trial from uniform distributions as $\alpha \sim \mathcal{U}(0,1)$ and $\beta \sim \mathcal{U}(\alpha, 1)$. In this uniform reward environment, the proposed myopic Whittle index policy maintains superior performance, achieving an average of $13\%$ reduction in regret to the other baseline algorithms at $T=4000$ steps. In particular, we can see that the baseline algorithms perform comparably to Round Robin, effectively behaving as non-adaptive strategies in the absence of distinct arm differences. This degradation is particularly evident in homogeneous settings with similar arm rewards, where standard policies deviate from asymptotic optimality bounds due to inadequate handling of restless state transitions in unselected arms, as discussed in~\cite{verloop2016asymptiotically}. Moreover, we noticed that the Whittle index policy perform the worst when the magnitude of the eigenvalue $|1 - \alpha - \beta|$ approaches zero, indicative of fast-mixing Markov chains where states exhibit rapid variability and minimal persistence from one time step to the next, thereby reducing the efficacy of the index policy in capturing transient dynamics. Despite this, the computational complexity of the Whittle index remains lower than that of the more sophisticated alternatives, making it preferable when the state transition parameters can be estimated or known a priori to the decision entity.

Figure \ref{fig:regretHeter_corollary3} compares the average strong cumulative regret over time for the Whittle index policies derived in Corollaries \ref{cor:whittle_discrete} and \ref{cor:whittle_discrete_vanishing} with $K = 50$ arms, each comprising $N = 10^5$ independent Markov processes. For each Monte Carlo trial, the transition parameters are sampled independently per arm as $\alpha_k \sim \mathcal{U}(0, 1)$ and $\beta_k \sim \mathcal{U}(\alpha_k, 1)$, introducing heterogeneity across arms and reflecting a broad spectrum of restless dynamics. To highlight the transient dynamics and the influence of the discount factor $\gamma$, the simulations initialize each Markov process uniformly at random in either state 0 or 1, rather than at the ergodic steady state, prompting the arms to converge toward their stationary distributions over time. This non-stationary starting point accentuates the behavior of $\gamma$, with the figure evaluating performance across a myopic regime ($\gamma=0$), where the policy prioritizes immediate rewards; a short-sighted regime ($\gamma=0.3$), where the policy balances immediate rewards with moderate consideration of future gains; and a farsighted regime ($\gamma=0.7$), where the policy assigns greater weight to delayed rewards and sustains more persistent exploration of arms. Overall, the analysis demonstrates that both corollaries produce nearly indistinguishable regret profiles in this limit across all discount factors, confirming the analytical consistency between the two formulations. Notably, the results reveal that lower values of $\gamma$ lead to substantially reduced cumulative regret, with the myopic policy achieving the lowest asymptotic regret due to its rapid adaptation to the unstable initial conditions and fast transition to steady state. Once the system stabilizes and the Whittle indices converge to the steady-state behavior of the system, the regret profiles exhibit parallel growth rates, as indicated by the linear segments of the curves. These findings suggest that, in environments with high arm heterogeneity and restless dynamics, myopic strategies emphasizing immediate information gain outperform more forward-looking approaches, though the tunable $\gamma$ allows adaptation to other scenarios where larger values may yield more favorable cumulative regrets.

Figure \ref{fig:regretHeter_corollary4} compares the average strong cumulative regret over time for the Whittle index policies from Corollaries \ref{cor:whittle_discrete} and \ref{cor:whittle_continous_vanishing}, representing respectively the discrete-time and continuous-time formulations of our framework. The setup features a heterogeneous configuration with $K = 50$ arms, each comprising $N = 10^5$ independent Markov processes, where transition parameters $\alpha_k \sim \mathcal{U}(0, 1)$ and $\beta_k \sim \mathcal{U}(\alpha_k, 1)$ vary per arm, and a fine time step $\Delta_t = 10^{-2}$ enables detailed modeling of rapid dynamics. Each process starts uniformly at random in state 0 or 1, diverging from the ergodic steady state to underscore transient effects and the impact of the discount factor $\gamma$, evaluated under a farsighted regime ($\gamma=0.7$), where the policy assigns greater weight to delayed rewards and sustains more persistent exploration of arms; a highly farsighted regime ($\gamma=0.9$), where the policy places strong emphasis on long-term reward accumulation; and an extremely farsighted regime ($\gamma=0.99$), where the policy further focuses on delayed rewards and approaches the limiting undiscounted case. The analysis confirms the consistency of both formulations across these settings, revealing that higher $\gamma$ values exacerbate regret due to challenges in predicting long-term rewards under the unstable initial conditions and the rapid state mixing from the unstable initial conditions. As the system stabilizes and Whittle indices settle into steady-state behavior, the cumulative regret curves exhibit parallel slopes, evident in the linear segments of the plots. The flexibility in tuning the discount factor for calculating the Whittle index proves advantageous, with lower values of $\gamma$ performing well in our fast-evolving simulation, while higher values may prove preferable in scenarios with slower or more predictable dynamics.

\section{Experimental results}
\label{sec:numerical_results}
To demonstrate the practical utility of our proposed RMPMAB framework and Whittle index policy, we apply it to a real-world biological imaging dataset focused on cell cycle profiling. This section is divided into three parts: the experimental setup for data acquisition, the image processing pipeline to adapt the data to the RMAB framework, and the experimental evaluation of the policy's performance.

\subsection{Data gathering}
U2OS cells expressing the FUCCI (Fluorescent Ubiquitination-based Cell Cycle Indicator) system were cultured in a Cytomat automated incubator (Thermo Fisher Scientific) at 37$^\circ$C  with $5\%$ CO$_2$ to maintain optimal growth conditions throughout the experiment. With that, we had direct visualization of cell cycle progression through distinct fluorescent markers: each nucleus glows red in G1 phase (mKO2-hCdt1, $561$ nm excitation), green in G2/M phases (mAG-hGeminin, $488$ nm excitation), or a blazing orange when both markers collide in S phase~\cite{sakaue2008visualizing}. Cells were seeded in a $96$-well plate and allowed to attach and proliferate overnight prior to imaging. All wells contained cells under identical culture conditions, allowing assessment of the natural heterogeneity and asynchrony in cell cycle progression across spatially separated populations.

Automated time-lapse imaging was performed using a Cephla Squid widefield microscope, scanning all $96$ wells sequentially at one-hour intervals over a $29$-hour period. For each well, the microscope focused on four predefined fields of view located centrally in the well to capture representative cellular activity while avoiding edge effects and minimizing imaging time per well. At each time point, three channels were acquired per field of view: brightfield for cell morphology and confluency assessment, $561$ nm excitation for red fluorescence, and $488$ nm excitation for green fluorescence. Images were captured at a resolution of $3000 \times 3000$ pixels with 10$\times$ magnification. The one-hour interval between complete plate scans was chosen to balance temporal resolution with practical constraints: each full scanning cycle required approximately $15$-$20$ minutes, and longer intervals minimized both photo-toxicity from repeated fluorescence excitation and temperature fluctuations from removing the plate from the incubator. This fully automated workflow eliminated manual intervention and ensured consistent imaging conditions across the entire experiment, yielding a total of $384$ images per channel per hour.



\subsection{Image processing}
To integrate the biological imaging data into the RMPMAB framework, we developed a processing pipeline with the primary goal of identifying the field of view with the highest number of cells in the G1 phase at each time step. Focusing on this objective, we extracted binary state information from the fluorescence images by processing only the $561$ nm channel.

In the resulting dataset, we have a total of $104$ images per well across all four fields of view over the entire experiment. For each field of view, we first computed a background image via pixel-wise median across all $29$ time points to capture static artifacts like uneven illumination or debris, subtracting it from each time-specific image to enhance signal-to-noise and isolate dynamic cellular signals. We then normalized intensities to $[0,1]$ and applied a threshold at the $98.5$th percentile, such that pixels above this threshold were assigned a value of $1$ (indicating potential G1 activity), while those below were set to $0$. For finer granularity suitable for RMPMAB modeling, we divided each binarized $3000\times3000$ pixel image into $40\times40$ pixel grids, yielding a total of $5625$ grid elements per field of view. Each grid was classified as active (state $1$) if more than $40\%$ of its pixels were $1$, otherwise it was classified as inactive (state $0$). This approach models grids as Markov processes, capturing localized G1 activity. Figure~\ref{fig:mask_overlay} illustrates this for a representative field of view, overlaying the $561$ nm binary mask (in red) on the brightfield image to highlight detected G1 regions.

In the RMAB formulation, the $96$ wells are mapped to $K=96$ arms, each comprising $N = 22500$ processes when aggregated across all four fields of view per well. The observable $Y_k(h)$ for arm $k$ at time step $h$ represents the percentage of active grids in the selected well, with the policy aiming to maximize G1 detections by prioritizing high-activity wells. Using the initial three hours of the $29$-hour dataset, we estimated transition parameters $\alpha_k$ and $\beta_k$ per well via maximum likelihood to parameterize the Whittle index policy for decision-making. For the remaining $26$ hours, we enhanced robustness through data augmentation by bootstrapping hourly observations, averaging combinations of two and three fields of view to generate $10$ points per well per hour to mimic partial-view variability. Figure~\ref{fig:coverage_heatmap} shows a heatmap of active grid percentages across all wells and hours, illustrating the spatio-temporal G1 activity patterns for policy evaluation.
\begin{figure*}[!t]
    \centering
    \subfigure[]{\includegraphics[width=0.40\textwidth]{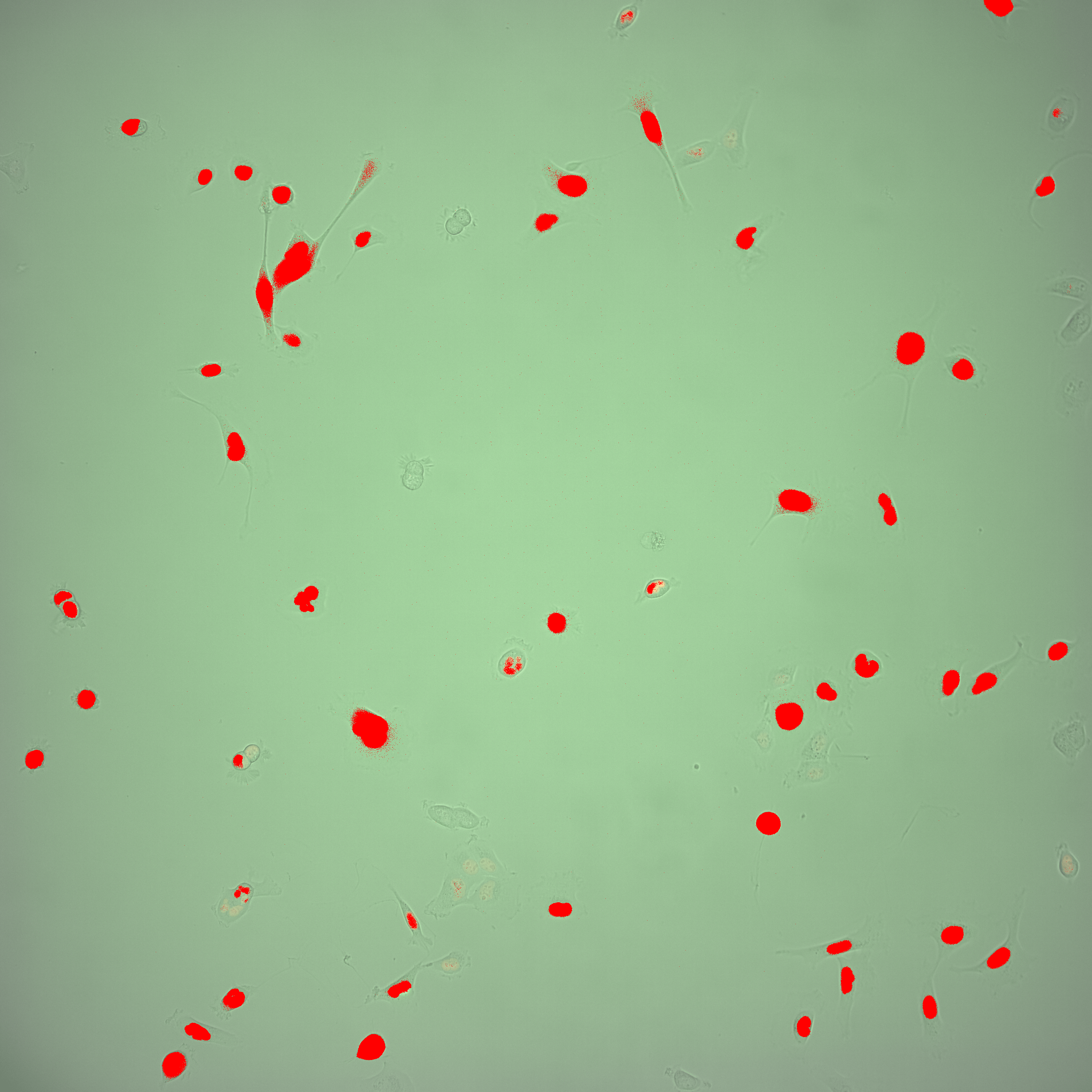} \label{fig:mask_overlay}}$\quad$
    \subfigure[]{\includegraphics[width=0.43\textwidth]{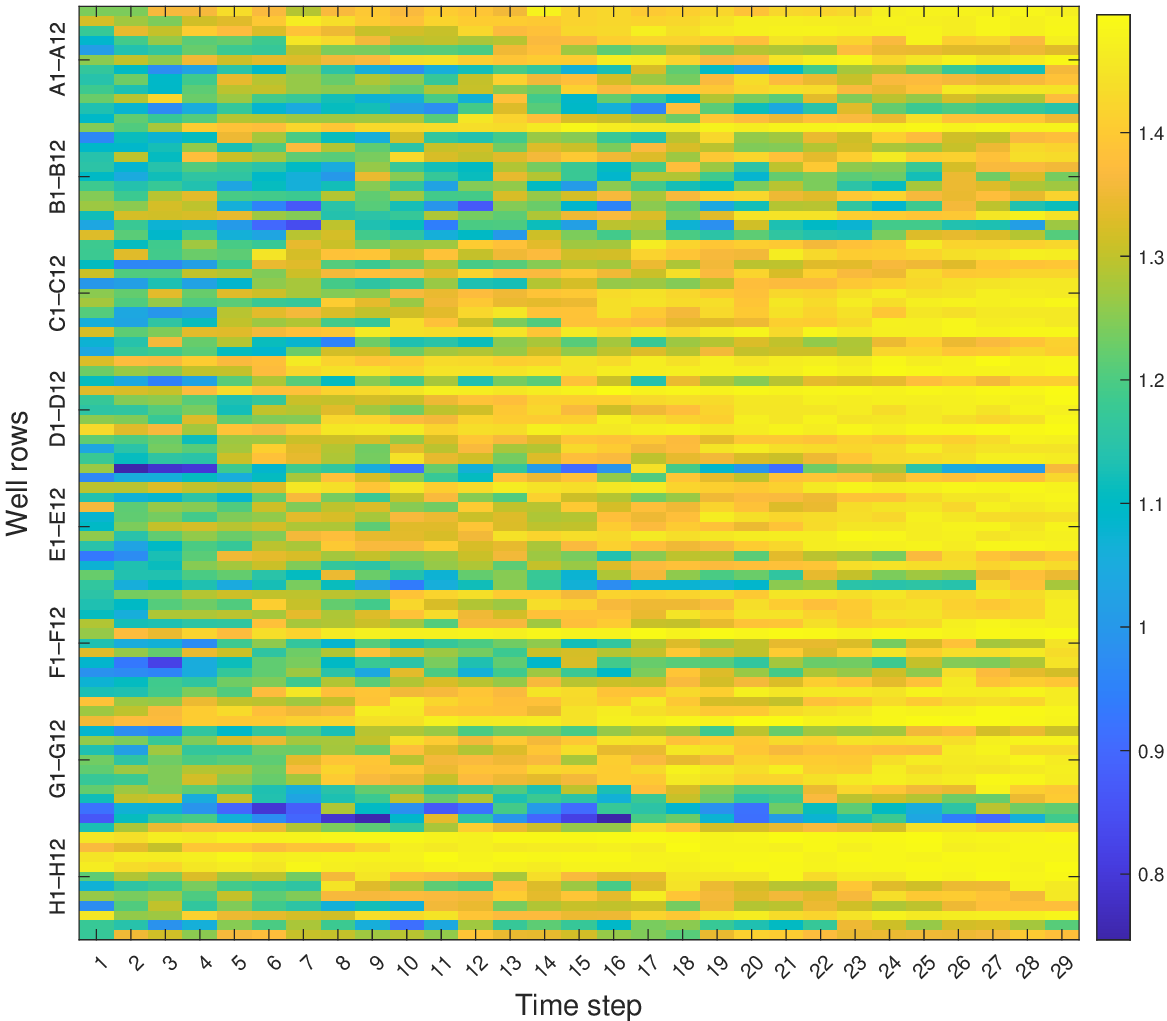} \label{fig:coverage_heatmap}}
    \caption{(a) Example of binary segmentation of G1-active regions (red overlay) on brightfield images for a representative well, after background subtraction and thresholding of the $561$ nm fluorescence channel. (b) Heatmap of the percentage of active grids across all wells and time hours, illustrating illustrating the spatio-temporal dynamics of G1 activity used for policy evaluation.}
    \label{fig:numerical_results}
\end{figure*}

\subsection{Experimental evaluation}
We evaluated the performance of the proposed Whittle index policy using the processed biological dataset, focusing on its ability to optimize the selection of wells with the highest G1 activity over time. Figure~\ref{fig:regret_biology} presents the cumulative strong regret over $240$ time steps, comparing the Whittle index policy with $\gamma = 0$, the Whittle index with $\gamma = 0.99$, Thompson Sampling, and the Round Robin strategy. Our analysis revealed that the activation parameter $\alpha$ varies across grid elements due to localized cellular dynamics, prompting the selection of the Whittle index policy as derived in Corollary III.3, which accounts for such heterogeneity.

The results demonstrate that the Whittle index policy significantly outperforms all other methods, achieving a $93\%$ reduction in cumulative regret compared to Round Robin at $T = 240$. Thompson Sampling follows as a competitive alternative, yet its regret is almost $10$ times larger than the Whittle policy, reflecting its suboptimal well selection due to limited adaptation to restless dynamics. The purely myopic Whittle index policy ($\gamma = 0$) underscores the benefit of incorporating future rewards, but it performes slightly worse than the farsighted Whittle index policy ($\gamma = 0.99$). The Round Robin method, which represents the current de facto standard of most labs, proves the least effective, highlighting the potential for substantial improvement with the proposed policy.
\begin{figure}
    \centering
    \includegraphics[width=0.5\linewidth]{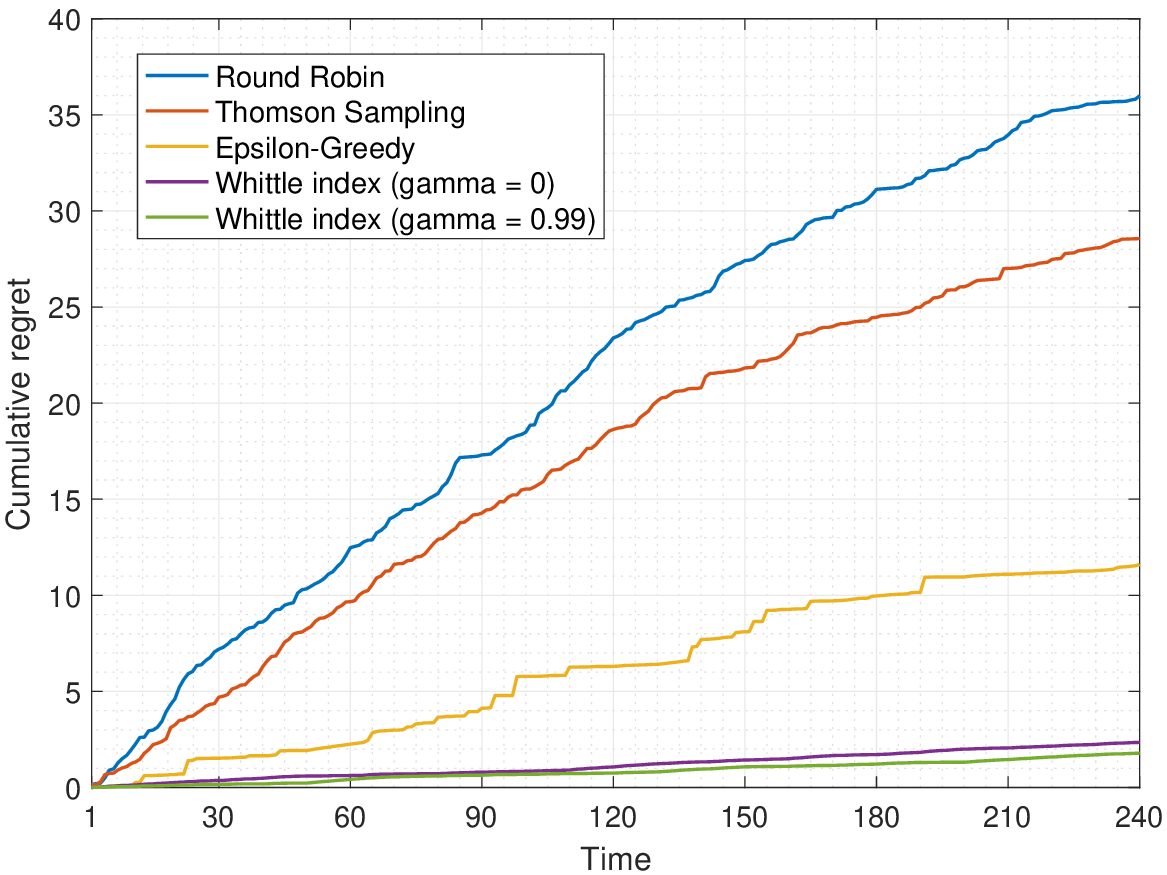}
    \caption{Cumulative strong regret over $240$ decision steps in the biological imaging dataset, comparing the proposed Whittle index policy ($\gamma = 0$ and $\gamma = 0.99$) with Round Robin, $\epsilon$-Greedy, and Thompson Sampling. Our proposed Whittle index policy achieves the lowest regret, demonstrating superior well selection under restless multi-process multi-arm dynamics.}
    \label{fig:regret_biology}
\end{figure}

\section{Conclusions}
\label{sec:conclusions}
In this paper, we introduced a novel formulation of restless multi-armed bandits tailored to systems comprising multiple independent Markov processes per arm, and we addressed the computational challenges inherent in large-scale decision-making under partial observability. By establishing a rigorous mathematical framework for the transient and asymptotic behaviors of single and aggregated Markov processes, we derived closed-form expressions and limiting approximations that enable scalable representations of system dynamics. Leveraging these insights, we developed a Whittle index policy based on minimal sufficient statistics, providing explicit forms for different discrete- and continuous-time settings. Numerical simulations validated the policy's superior performance over established baselines such as Round Robin, Bayesian UCB, and Thompson Sampling, exhibiting also a cumulative regret that scaled favorably under varying number of arms and processes with more than a $37\%$ reduction in cumulative regret. Experimental results were also performed on a biological microscopy dataset, modeling cellular wells as arms with transient activation events (fluorescence signals) drawn from high-throughput imaging experiments. Our proposed index policy in this setup demonstrated substantial improvements, capturing $93\%$ more biologically relevant events compared to cycling through the wells, highlighting the policy's efficacy in resource-constrained settings where the microscope can only image one well at a time.

Overall, our contributions advance the state of the art in multi-armed bandit problems by extending the Whittle index to multi-process arms, offering a principled balance between optimality and tractability in restless environments. The framework's applicability extends beyond smart microscopy to resource-constrained scenarios in experimental sciences, such as adaptive sampling in high-throughput drug screening or real-time monitoring in dynamic networks. Future work could explore the regret bounds for the proposed policy, the integration with deep reinforcement learning for non-Markovian extensions, or the empirical validations in additional domains like cybersecurity or renewable energy management, where the study of aggregated stochastic processes is crucial for optimal decision-making.

\bibliography{references}
\bibliographystyle{IEEEtran}

\appendix
\subsection{Statistical properties of infinitely many independent Markov processes}
\label{app:proof-TranProb_infiniteProcesses}
Given $Y(h)=j$, there are exactly $j$ processes in state $1$ and $N-j$ processes in state $0$ at time $h$. Since the processes are independent and identically distributed, we can analyze their collective behavior by considering the transitions of these two groups separately over $m$ time steps. Define $Y^{\mathsf{AA}}(h+m)$ as the number of processes that are in state $1$ at time $h$ and remain in state $1$ at time $h+m$, and define $Y^{\mathsf{IA}}(h+m)$ as the number of processes that are in state $0$ at time $h$ and transition to state $1$ by time $h+m$. The total number of active processes at time $h+m$ is then
\begin{equation}
    Y(h+m) = Y^{\mathsf{AA}}(h+m) + Y^{\mathsf{IA}}(h+m).
\end{equation}

Since all the processes are independent, $Y^{\mathsf{AA}}(h+m)$ and $Y^{\mathsf{IA}}(h+m)$ are independent random variables. The distribution of $Y^{\mathsf{AA}}(h+m)$ follows from the 
$j$ processes starting in state $1$, each of which remains in state $1$ after $m$ steps with probability
\begin{equation*}
    q_m = [\bm{P}^m]_{11} = \frac{1}{\alpha+\beta}\Big[ \alpha + \beta(1-\alpha-\beta)^m \Big], \qquad \forall m\geq0.
\end{equation*}
Similarly, $Y^{\mathsf{IA}}(h+m)$ arises from the $N-j$ processes starting in state $0$, each transitioning to state $1$ after $m$ steps with probability
\begin{equation*}
    p_m = [\bm{P}^m]_{10} = \frac{1}{\alpha+\beta}\Big[ \alpha- \alpha (1-\alpha-\beta)^m \Big], \qquad \forall m\geq0.
\end{equation*}
Hence, $Y^{\mathsf{AA}}(h+m)$ and $Y^{\mathsf{IA}}(h+m)$ are both independent binomial random variables, each defined as
\begin{equation}
    Y^{\mathsf{AA}}(h+m)\mid Y(h)=j \sim \text{Binomial}(j,q_m),\qquad \text{and} \qquad Y^{\mathsf{IA}}(h+m)\mid Y(h)=j \sim \text{Binomial}(N-j,p_m).
\end{equation}

Now, setting $\alpha=\bar{\alpha}/N$ and substituting into the $m$-step transition probabilities, we obtain
\begin{align}
    q_m & = \frac{\frac{\bar{\alpha}}{N}}{\frac{\bar{\alpha}}{N} + \beta} + \frac{\beta}{\frac{\bar{\alpha}}{N} + \beta}\left(1 - \frac{\bar{\alpha}}{N} -  \beta \right)^m, \\
    p_m & = \frac{\frac{\bar{\alpha}}{N}}{\frac{\bar{\alpha}}{N} + \beta} \left[1 - \left(1 - \frac{\bar{\alpha}}{N} -  \beta \right)^m\right].
\end{align}
Using the geometric series formula $\frac{1}{1-x} = \sum_{n=0}^\infty x^n$, the inverse of the denominator $\frac{\bar{\alpha}}{N} + \beta = \beta\left(1+\frac{\bar{\alpha}}{N\beta}\right)$ approximates to
\begin{equation*}
     \frac{1}{\frac{\bar{\alpha}}{N} + \beta} = \frac{1}{\beta}\left(1 - \frac{\bar{\alpha}}{N\beta} +\mathcal{O}\left(\frac{1}{N^2}\right)\right).
\end{equation*}
Moreover, using the binomial expansion $(1-x)^m = 1-mx+\mathcal{O}(x^2)$ for $x=\frac{\bar{\alpha}}{N(1-\beta)}$, the exponent $\left(1 - \frac{\bar{\alpha}}{N} -  \beta \right)^m$ in both transition probabilities approximates to
\begin{equation*}
    \left(1 - \frac{\bar{\alpha}}{N} -  \beta \right)^m = (1-\beta)^m\left(1 - \frac{\bar{\alpha}}{N(1-\beta)}\right)^m = (1-\beta)^m\left(1 - \frac{m\bar{\alpha}}{N(1-\beta)} + \mathcal{O}\left(\frac{1}{N^2}\right)\right)
\end{equation*}
Putting all together, we have that,
\begin{equation*}
    q_m = \left[\frac{\bar{\alpha}}{N\beta}  + (1-\beta)^m\left(1 - \frac{m\bar{\alpha}}{N(1-\beta)} + \mathcal{O}\left(\frac{1}{N^2}\right)\right)\right]\left(1 - \frac{\bar{\alpha}}{N\beta} + \mathcal{O}\left(\frac{1}{N^2}\right)\right) = (1-\beta)^m + \mathcal{O}\left(\frac{1}{N}\right).
\end{equation*}
and
\begin{equation*}
    p_m = \frac{\bar{\alpha}}{N\beta} \left(1 - \frac{\bar{\alpha}}{N\beta} +\mathcal{O}\left(\frac{1}{N^2}\right)\right)\left(1-(1-\beta)^m + \mathcal{O}\left(\frac{1}{N}\right)\right) = \frac{\bar{\alpha}}{N\beta}(1-(1-\beta)^m) + \mathcal{O}\left(\frac{1}{N^2}\right).
\end{equation*}

For large $N$ and fixed $m$, observe that the higher-order terms in $q_m$ vanish and the transition probability of remaining active converges to $(1-\beta)^m$, whereas $p_m$ scales as $\mathcal{O}\left(\frac{1}{N}\right)$, and $Np_m\to\frac{\bar{\alpha}}{\beta}[1-(1-\beta)^m]$. Considering this, as $N\to\infty$, the binomial term $Y^{\mathsf{AA}}(h+m)$ converges in distribution to that of a binomial random variable,
\begin{equation}
    Y^{\mathsf{AA}}(h+m)\mid Y(h)=j \sim \text{Binomial}(j,q_m) \; \overset{d}{\longrightarrow} \; Z \sim \text{Binomial}\left(j,(1-\beta)^m\right),
    \label{eq:conDist_AA}
\end{equation}
and, using Le Cam's theorem~\cite{le1960approximation}, the binomial term $Y^{\mathsf{IA}}(h+m)$ converges in distribution to that of a Poisson random variable,
\begin{equation}
    Y^{\mathsf{IA}}(h+m)\mid Y(h)=j \sim \text{Binomial}(N-j,p_m) \; \overset{d}{\longrightarrow} \; W \sim \text{Poisson}\left(\frac{\bar{\alpha}}{\beta}\left(1-(1-\beta)^m\right)\right).
    \label{eq:conDist_IA}
\end{equation}
Moreover, since $Y^{\mathsf{AA}}(h+m)$ and $Y^{\mathsf{IA}}(h+m)$ are independent, so are their limits $Z$ and $W$. This establishes the distributional convergence stated in Lemma \ref{lemma:convergence_nrof_chains}. 

Now we focus on the expected value. Since $Y^{\mathsf{AA}}(h+m)$ and $Y^{\mathsf{IA}}(h+m)$ are independent, and their distributions converge as show in \eqref{eq:conDist_AA} and \eqref{eq:conDist_IA}, the limit of the conditional expectation can be expressed as
\begin{align}
    \mathbb{E}\big[Y(h+m)  \mid Y(h) = j\big] & = \mathbb{E}\big[Y^{\mathsf{AA}}(h+m)  \mid Y(h) = j\big] + \mathbb{E}\big[Y^{\mathsf{IA}}(h+m)  \mid Y(h) = j\big] = j q_m + (N-j)p_m \nonumber \\ 
    &= j \left( (1 - \beta)^m + \mathcal{O}\left(\frac{1}{N}\right) \right) + (N - j) \left( \frac{\bar{\alpha}}{N \beta} \left(1 - (1 - \beta)^m\right) + \mathcal{O}\left(\frac{1}{N^2}\right) \right).
\end{align}
where we used that the expectation of a binomial random variable is $\mathbb{E}[\text{Binomial}(n,p)]=np$. For the first term, since $j$ is finite, $j\mathcal{O}\left(\frac{1}{N}\right)\rightarrow 0$ as $N\rightarrow\infty$, so $jq_m\rightarrow j(1-\beta)^m$. Similarly, for the second term, since $j$ is fixed, $(N-j)/N\rightarrow 1$ as $N\rightarrow\infty$, and the higher-order terms $(N-j)\mathcal{O}\left(\frac{1}{N^2}\right) = \mathcal{O}\left(\frac{1}{N}\right) \rightarrow 0$. Putting all together, we obtain that the conditional expectation in the limit is
\begin{equation*}
    \lim_{N\rightarrow\infty} \mathbb{E}\big[Y(h+m)  \mid Y(h) = j\big] = j(1-\beta)^m + \frac{\bar{\alpha}}{\beta}\Big(1-(1-\beta)^m\Big).
\end{equation*}
This concludes the proof. \hfill $\blacksquare$

\subsection{Statistical properties of infinitely many independent continuous-time Markov processes}
\label{app:proof-TranProb_CTMC}
The transition probability matrix $\bm{P}$ of the discrete-time Markov processes reflects the directions of the transitions between states but does not account for the holding time in each of these states. To bridge the discrete-time framework to the continuous-time limit, let $\bm{B}(t)$ denote the continuous-time transition probability matrix, satisfying the Chapman-Kolmogorov equations
\begin{equation}
    \frac{\mathrm{d}}{\mathrm{d}t} \bm{B}(t) = \bm{Q} \bm{B}(t), \quad \text{with} \quad \bm{B}(0) = \bm{I},
\end{equation}
where $\bm{Q}$ is the generator matrix, and $\bm{I}$ is the identity matrix. The solution of this linear differential equation is given by the matrix exponential
\begin{equation}
    \bm{B}(t) = e^{\bm{Q}t}.
    \label{eq:matrix_exponentiation}
\end{equation}
If we further consider that the discrete-time process approximates the continuous-time evolution via the time step $\Delta_t>0$, such that any discrete time step $m\geq0$ corresponds to time $t=m\Delta_t$, the transition matrix $\bm{P}$ relates to the generator matrix $\bm{Q}$ over the infinitesimal time interval $\Delta_t$ via first-order Taylor expansion of the matrix exponential
\begin{equation}
    \bm{P} = e^{\bm{Q}\Delta_t} = \bm{I} + \bm{Q}\Delta_t + \mathcal{O}(\Delta_t^2),
\end{equation}
and the generator matrix $\bm{Q}$ may be recovered via
\begin{equation}
    \bm{Q} = \lim_{\Delta_t\to 0} \; \frac{1}{\Delta_t} \left[\bm{P} - \bm{I} \right].
\end{equation}
This formalism justifies the approximation of a continuous-time Markov process by a discrete-time process with transition probabilities scaled by $\Delta_t$. Consequently, adopting the parametrization $\alpha=\bar{\alpha}\Delta_t/N$ and $\beta = \bar{\beta}\Delta_t$ yields that
\begin{equation}
    \bm{Q} = \begin{bmatrix}
        -\bar{\alpha}/N & \bar{\beta} \\
        \bar{\alpha}/N & -\bar{\beta}
    \end{bmatrix}.
    \label{eq:generator_matrix}
\end{equation}

Given the diagonalization $\bm{Q}=\bm{V}\bm{\Lambda}\bm{V}^{-1}$, the state transition probability over a time period $\tau\geq 0$, defined via the matrix exponentiation in \eqref{eq:matrix_exponentiation}, can be obtained as
\begin{equation}
    \bm{B}(\tau) = e^{\bm{Q}\tau} = \bm{V}e^{\bm{\Lambda}\tau}\bm{V}^{-1} = \frac{1}{\frac{\bar{\alpha}}{N} + \bar{\beta}} \begin{bmatrix}
    \bar{\beta} & -1  \\
    \frac{\bar{\alpha}}{N} & 1
    \end{bmatrix}
    \begin{bmatrix}
    1 & 0 \\
    0 & e^{-\left(\frac{\bar{\alpha}}{N} + \bar{\beta}\right)}
    \end{bmatrix}
    \begin{bmatrix}
    1 & 1 \\
    -\frac{\bar{\alpha}}{N} & \bar{\beta}
    \end{bmatrix}.
\end{equation}
Evaluating its entries yields our two probabilities of interest, namely,
\begin{equation}
    q_\tau = \left[ \bm{B}(\tau) \right]_{11} = e^{-\bar{\beta}\tau} + \mathcal{O}\left(\frac{1}{N}\right), \quad \text{and} \quad p_\tau = \left[ \bm{B}(\tau) \right]_{10} = \frac{\bar{\alpha}}{N\beta}\left(1-e^{-\bar{\beta}\tau}\right) + \mathcal{O}\left(\frac{1}{N^2}\right).
    \label{eq:expansion(t)ransProb}
\end{equation}

We now turn our attention to characterizing the conditional distribution $Y(t+\tau)\mid Y(t)=j$. Given $Y(t)=j$, there are exactly $j$ processes in state $1$ and $N-j$ processes in state $0$ at time $t$. As such, we can express $Y(t+\tau)$ as the sum of two independent components
\begin{equation}
    Y(t+\tau) = Y^{\mathsf{AA}}(t+\tau) + Y^{\mathsf{IA}}(t+\tau), 
\end{equation}
where $Y^{\mathsf{AA}}(t+\tau)$ represents the number of initially active processes at time $t$ that remain active at time $t+\tau$, and $Y^{\mathsf{IA}}(t+\tau)$ represents the number of initially inactive processes at time $t$ that become active at time $t+\tau$. Given the independence of the processes, $Y^{\mathsf{AA}}(t+\tau)$ and $Y^{\mathsf{IA}}(t+\tau)$ are independent random variables with conditional probability distributions
\begin{equation}
    Y^{\mathsf{AA}}(t+\tau)\mid Y(t)=j \sim \text{Binomial}(j,q_\tau),\qquad \text{and} \qquad Y^{\mathsf{IA}}(t+\tau)\mid Y(t)=j  \sim \text{Binomial}(N-j,p_\tau).
\end{equation}

Considering the above, define the probability generating function (PGF) of $Y(t+\tau)$ conditioned on $Y(t)=j$ as
\begin{equation}
    G_\tau(s) = \mathbb{E}\bigl[s^{\,Y(t+\tau)}\mid Y(t)=j\bigr] = \bigl[1+(s-1)q_\tau\bigr]^j \bigl[1+(s-1)p_\tau\bigr]^{N-j},
    \label{eq:PGF_conditional}
\end{equation}
where we used that the PGF of a bernoulli random varible $\text{Binomial}(n,p)$ is $g(s)=\bigl[1+(s-1)p\bigr]^n$. Using the expansions for $q_\tau$ and $p_\tau$ in \eqref{eq:expansion(t)ransProb}, we now proceed to analyze each factor in the limit $N\to\infty$. For the first term,
\begin{equation}
    \lim_{N\to\infty} \bigl[1+(s-1)q_\tau\bigr]^j = \lim_{N\to\infty} \left[1+(s-1)\left(e^{-\bar{\beta}\tau} + \mathcal{O}\left(\frac{1}{N}\right)\right)\right]^j = \left(1+(s-1) e^{-\bar{\beta}\tau}\right)^j,
\end{equation}
which corresponds to the PGF of a $\text{Binomial}\left(j, e^{-\bar{\beta}\tau}\right)$ random variable. For the second term, consider
\begin{equation}
    \bigl[1+(s-1)p_\tau\bigr]^{N-j}  = \exp\left[(N-j)\log\left[1+(s-1)\left( \frac{\bar{\alpha}}{N\beta}\left(1-e^{-\bar{\beta}\tau}\right) + \mathcal{O}\left(\frac{1}{N^2}\right) \right)\right]\right].
\end{equation}
Using the logarithmic expansion $\log(1+x)=\sum_{n=1}^\infty (-1)^{n+1} x^n/n$ for $x=(s-1)\frac{\bar{\alpha}}{N\beta}(1-e^{-\bar{\beta}\tau})$, we have
\begin{equation}
    \bigl[1+(s-1)p_\tau\bigr]^{N-j} = \exp\left[(N - j) \left( (s-1) \frac{\bar{\alpha}}{N \beta} \left(1 - e^{-\bar{\beta}\tau}\right) + \mathcal{O}\left(\frac{1}{N^2}\right) \right)\right],
\end{equation}
such that, in the limit,
\begin{equation}
    \lim_{N\to\infty} \bigl[1+(s-1)p_\tau\bigr]^{N-j} = \exp\left[(s-1) \frac{\bar{\alpha}}{\beta} \left(1 - e^{-\bar{\beta}\tau}\right)\right],
\end{equation}
which corresponds to the PGF of a $\text{Poisson}\left(\frac{\bar{\alpha}}{\beta} (1 - e^{-\bar{\beta}\tau})\right)$ random variable.

Combining the two limits, the point-wise limit of the PGF in \eqref{eq:PGF_conditional} for fixed $s$ is
\begin{equation}
    \lim_{N \to \infty} G_\tau(s) = \underbrace{\left(1+(s-1) e^{-\bar{\beta}\tau}\right)^j}_{\text{Binomial term}} \; \underbrace{\exp\left[(s-1) \frac{\bar{\alpha}}{\beta} \left(1 - e^{-\bar{\beta}\tau}\right)\right]}_{\text{Poisson term}}.
    \label{eq:limit_PGF}
\end{equation}

By Lévy's continuity theorem~\cite{williams1991probability}, we can conclude that $Y(t+\tau)\mid Y(t)=j$ converge in distribution to the sum of two independent random variables $Z$ and $W$, each corresponding to the characteristic functions of $Y^{\mathsf{AA}}(t+\tau)\mid Y(t)=j$ and $Y^{\mathsf{IA}}(t+\tau)\mid Y(t)=j$ in the limit, with
\begin{align}
    Y^{\mathsf{AA}}(t+\tau)\mid Y(t)=j & \sim \text{Binomial}(j,q_\tau) \; \overset{d}{\longrightarrow} \; Z \sim \text{Binomial}\left(j, e^{-\bar{\beta}\tau}\right),\\[.3cm]
    Y^{\mathsf{IA}}(t+\tau)\mid Y(t)=j & \sim \text{Binomial}(j,p_\tau) \; \overset{d}{\longrightarrow} \; W \sim \text{Poisson}\left(\frac{\bar{\alpha}}{\beta} \left(1 - e^{-\bar{\beta}\tau}\right)\right).
\end{align}
This establishes the distributional convergence stated in Lemma \ref{lemma:convergence_nrof_chains_CTMC}.

Finally, by differentiating the limit PGF in \eqref{eq:limit_PGF} with respect to $s$ and evaluating at $s=1$, we obtain the limiting conditional expectation
\begin{equation}
    \lim_{N\to\infty} \mathbb{E}\left[Y(t+\tau) \mid Y(t) = j\right] = \frac{\partial}{\partial s}\lim_{N\to\infty} G_\tau(s)\Big|_{s=1} = \frac{\bar{\alpha}}{\bar{\beta}}\Big(1-e^{-\bar{\beta}\tau}\Big) + je^{-\bar{\beta}\tau}.
\end{equation}
This concludes the proof.  \hfill $\blacksquare$

\subsection{Whittle index}
\label{app:proof-Whittle-Index}
The restless multi-armed bandit problem with subsidy $\lambda$ for passivity is a discounted Markov decision problem with discount factor $\gamma \in (0, 1)$, where the value function of any arm $k$ in state $(j_k, m_k)$ satisfies the Bellman equation
\begin{equation}
     V (j_k, m_k; \lambda, \gamma) = \max \Big\{  V^\mathcal{P}(j_k, m_k; \lambda, \gamma),  V^\mathcal{A}(j_k, m_k; \lambda, \gamma)\Big\}.
\end{equation}
Here $V^\mathcal{P}$ and $V^\mathcal{A}$ represent the value functions under passive and active actions, respectively, and are defined as
\begin{align}
    V^\mathcal{P}(j_k, m_k; \lambda, \gamma) &= \lambda + \gamma V (j_k,m_k+1;\lambda,\gamma ) \label{eq:passive-valueFunc} \\
    V^\mathcal{A}(j_k, m_k; \lambda, \gamma) &= \mathbb{E}\left[Y_k\mid j_k, m_k\right] + \gamma \sum_{y=0}^{N_k} \mathcal{P}(Y_k = y\mid j_k,m_k)V (y,0;\lambda,\gamma). \label{eq:active-valueFunc}
\end{align}
The passive value function represents the case where arm $k$ is not selected, and it captures the immediate subsidy $\lambda$ plus the discounted future reward when the arm remains unobserved for one additional time step. Contrarily, the active value function represents the case where arm $k$ is selected, and it comprises the expected immediate reward from the observation plus the discounted expected future reward after resetting the belief state.

At this stage, we must establish that the problem is indexable. The indexability condition ensures that the Whittle index is well-defined and provides a principled approach to decomposing the multi-armed problem into tractable single-arm subproblems~\cite{whittle1988restless}. For that, define the passive set for arm $k$ as the set of states where the passive action is optimal under subsidy $\lambda$,
\begin{equation}
    \Omega_k(\lambda,\gamma) = \Big\{(j_k,m_k)\in \{0,1,\dots,N_k\}\times\mathbb{N}_0 : V^\mathcal{P}(j_k, m_k; \lambda, \gamma)\geq  V^\mathcal{A}(j_k, m_k; \lambda, \gamma)\Big\}.
    \label{eq:passive-set-indexability}
\end{equation}
Indexability requires that $\Omega_k(\lambda,\gamma)$ is monotone non-decreasing in $\lambda$. We note that our single-arm model falls into the subclass of controlled‐restart restless bandits, where taking the active action resets the state according to a known probability distribution (given explicitly in \eqref{eq:transProbMatrix_multipleProcesses}). This model satisfies the sufficient conditions for indexability conditions established in~\cite[Theorem~1]{akbarzadeh2022conditions}. Furthermore, the partially observable variant of our model, in which the state is revealed only when activating an arm, also falls within the indexable families characterized in~\cite[Theorem~1]{akbarzadeh2022partially}, who prove monotonicity of the passive set for partially observable restart bandits. Therefore, since both structural and informational assumptions of our model satisfy the criteria in these works, we omit a full monotonicity proof and rely on the established result.

Given indexability, the Whittle index $\mathcal{W}(j_k,m_k;\gamma)$ can be defined for any arm $k\in\mathcal{K}$ and for any discount factor $\gamma\in(0,1)$ following its theoretical characterization in~\cite[Theorem 4.3]{beutler1985optimal}  as the smallest subsidy $\lambda$ for which the state $(j_k,m_k)$ is part of the passive set,
\begin{equation}
    \mathcal{W}(j_k,m_k;\gamma) = \inf \Big\{\lambda\in\mathbb R : (j_k, m_k)\in \Omega_k(\lambda,\gamma) \Big\}.
\end{equation}
That is, at the Whittle index $\lambda^* = \mathcal{W}(j_k, m_k; \gamma)$, the decision entity is indifferent between making an arm active or leaving it passive. Considering this, and given the linearity of the conditional expectation $\mathbb{E}[Y_k \mid j_k, m_k] = C(m_k) + j_k D(m_k)$ and the recursive structure of the problem, the value function admits an affine representation in $j_k$ at this indifference point,
\begin{equation}
    V^\mathcal{P}(j_k, m_k; \lambda^*, \gamma) = V^\mathcal{A}(j_k, m_k; \lambda^*, \gamma) = A(m_k; \gamma) + j_k \, B(m_k ; \gamma),
    \label{eq:affine_form}
\end{equation}
where $A(m_k; \gamma) $ and $B(m_k; \gamma) $ are functions that depend only on the delay $m_k$ and the discount factor $\gamma$. To verify consistency with the Bellman recursion, substitute the affine form of the future value function and the expression for the expected reward  into the passive and active value functions in \eqref{eq:passive-valueFunc} and \eqref{eq:active-valueFunc},
\begin{align}
    V^\mathcal{P}(j_k, m_k; \lambda^*, \gamma) &= \lambda^* + \gamma A(m_k+1; \gamma) + j_k \,\gamma B(m_k+1 ; \gamma), \\
    V^\mathcal{A}(j_k, m_k; \lambda^*, \gamma) &= \gamma A(0; \gamma) + C(m_k) \big( 1+ \gamma B(0; \gamma)\big)+ j_k \, D(m_k) \big( 1+ \gamma B(0; \gamma)\big).
\end{align}
Both expressions are indeed affine in $j_k$, and its form is preserved under recursion. Note, however, that the affine representation in~\eqref{eq:affine_form} is only valid at the indifference subsidy $\lambda^*$, which is the only point needed to define the index. Outside indifference, the value function is the maximum of two affine functions and hence piecewise affine.

From there, equating the expressions for $V^\mathcal{P}$ and $V^\mathcal{A}$ to the affine form in \eqref{eq:affine_form} yields two key insights. First, the Whittle index must satisfy the condition
\begin{equation}
    \mathcal{W}(j_k,m_k;\gamma) = A(m_k; \gamma) - \gamma A(m_k+1; \gamma) + j_k \,( B(m_k; \gamma) - \gamma B(m_k+1 ; \gamma)),
    \label{eq:lambda_indifference}
\end{equation}
and second, the functions $A(m_k; \gamma)$ and $B(m_k; \gamma)$ must satisfy the recursions
\begin{align}
    A(m_k; \gamma) &= \gamma A(0; \gamma) + C(m_k) \big( 1+ \gamma B(0;\gamma)\big) \label{eq:recursion_A}\\
    B(m_k; \gamma) &= D(m_k) \big( 1+ \gamma B(0;\gamma)\big).  \label{eq:recursion_B}
\end{align}

Setting $m_k = 0$ into the recursions and solving the equations yield
\begin{equation}
    A(0; \gamma) = \frac{C(0)}{(1-\gamma)(1-\gamma D(0))}\quad\quad \text{and} \quad\quad  B(0; \gamma) = \frac{D(0)}{1-\gamma D(0)}.
    \label{eq:initial}
\end{equation}
Then, substituting these into the recursions provides explicit expressions for all $m_k$,
\begin{equation}
     A(m_k; \gamma) = \frac{\gamma C(0)}{(1-\gamma)(1-\gamma D(0))} + \frac{C(m_k)}{1-\gamma D(0)}\quad\quad \text{and} \quad\quad B(m_k; \gamma) = \frac{D(m_k)}{1-\gamma D(0)}.
\end{equation}
Finally, substituting the above expressions into \eqref{eq:lambda_indifference} and simplifying gives the closed-form expression for the Whittle index:
\begin{equation*}
    \mathcal{W}(j_k,m_k;\gamma) = \frac{1}{1-\gamma D(0)}\Big[ C(m_k) - \gamma C(m_k+1) +\gamma C(0)\Big] + \frac{j_k}{1-\gamma D(0)}\Big[D(m_k)-\gamma D(m_k+1)\Big].
\end{equation*}

This concludes the proof.  \hfill $\blacksquare$

\end{document}